\documentclass{book} %

\usepackage[utf8]{inputenc}
\usepackage{cite}
\usepackage{amsfonts}
\usepackage{amsmath}
\usepackage{amssymb}
\usepackage{acronym} %
\usepackage{booktabs} %
\usepackage{enumitem} %
\usepackage{subcaption}
\usepackage{tikz}
\tikzset{>=stealth}
\usetikzlibrary{intersections,calc,through,positioning}
\tikzset{every picture/.style={line width=1pt}}
\usepackage{overpic}
\usepackage{xcolor}
\usepackage{framed}
\usepackage{newunicodechar} %
\newunicodechar{ℓ}{\ensuremath{\ell}}
\newunicodechar{₁}{\ensuremath{_1}}
\newunicodechar{₂}{\ensuremath{_2}}
\usepackage{pgfplots} %
\pgfplotsset{compat=1.14}
\usepackage{hyperref}
\usepackage{amsthm}
\RequirePackage[vcentering,hcentering]{geometry}
\definecolor{shadecolor}{rgb}{.9,.9,.9}

\makeatletter
\AtBeginDocument{%
  \renewcommand*{\AC@hyperlink}[2]{%
    \begingroup
      \hypersetup{hidelinks}%
      \hyperlink{#1}{#2}%
    \endgroup
  }%
}
\makeatother

\newtheorem{theorem}{Theorem}
\newtheorem{definition}{Definition}
\newtheorem{example}{Example}
\newtheorem{remark}{Remark}

\newcommand{\vv}[1]{{\boldsymbol{\mathbf{#1}}}}
\newcommand{\iv}[1]{{\boldsymbol{#1}}}
\newcommand{\ip}[2]{\left\langle#1,#2\right\rangle}

\renewcommand{\imath}{\textrm{j}}
\newcommand{\e}{\textrm{e}}

\newcommand{\samp}{S}
\newcommand{\fourier}{\mathcal{F}}
\newcommand{\mult}{M}
\newcommand{\conv}{C}
\newcommand{\COV}{\Phi}
\newcommand{\synth}{T}
\newcommand{\xray}{X}

\DeclareMathOperator{\rect}{rect}

\DeclareMathOperator{\sinc}{sinc}

\DeclareMathOperator{\prox}{prox}

\DeclareMathOperator*{\argmin}{arg\,min}
\DeclareMathOperator*{\argmax}{arg\,max}

\title{Biomedical Image Reconstruction: From the Foundations to Deep Neural Networks}

\author{
Michael T. McCann \\
\and
Michael Unser
}
\date{École polytechnique fédérale de Lausanne}

\begin{document}

\maketitle

\chapter*{Abstract}
  This tutorial covers biomedical image reconstruction,
  from the foundational concepts of system modeling and direct reconstruction
  to modern sparsity and learning-based approaches.
  
  Imaging is a critical tool in biological research and medicine,
  and most imaging systems necessarily use an image-reconstruction algorithm
  to create an image;
  the design of these algorithms has been a topic of research since at least the 1960's.  
  In the last few years,
  machine learning-based approaches have shown impressive performance on image reconstruction problems,
  triggering a wave of enthusiasm and creativity around the paradigm of learning.
  Our goal is to unify this body of research,
  identifying common principles and reusable building blocks across
  decades and among diverse imaging modalities.

  We first describe system modeling, emphasizing
  how a few building blocks can be used to describe a broad range of imaging modalities.
  We then discuss reconstruction algorithms, grouping them into three broad generations.
  The first are the classical direct methods, including Tikhonov regularization;
  the second are the variational methods based on sparsity and the theory of compressive sensing;
  and the third are the learning-based (also called data-driven) methods,
  especially those using deep convolutional neural networks.
  There are strong links between these generations:
  classical (first-generation) methods appear as modules inside the latter two,
  and the former two are used to inspire new designs for learning-based (third-generation) methods.
  As a result,
  a solid understanding of all three generations is necessary for the design of state-of-the-art algorithms.

\chapter*{List of Abbreviations}
{\protect \addcontentsline{toc}{chapter}{List of Abbreviations}}%
\begin{acronym}
  \acro{ADMM}[ADMM]{alternating direction method of multipliers}
  \acro{CCD}[CCD]{charge-coupled device}
  \acro{CG}[CG]{conjugate gradient}
  \acro{CNN}[CNN]{convolutional neural network}
  \acro{CT}[CT]{computed tomography}
  \acro{DCT}[DCT]{discrete cosine transform}
  \acro{ET}[ET]{electron tomography}
  \acro{FBP}[FBP]{filtered back projection}
  \acro{FFT}[FFT]{fast Fourier transform}
  \acro{GPU}[GPU]{graphics processing unit}
  \acro{iid}[i.i.d.]{independent and identically distributed}
  \acro{ISTA}[ISTA]{iterative shrinkage and thresholding}
  \acro{MAP}[MAP]{maximum a posteriori}
  \acro{MMSE}[MMSE]{minimum mean square error}
  \acro{MRI}[MRI]{magnetic resonance imaging}
  \acro{MSE}[MSE]{mean squared error}
  \acro{PDF}[PDF]{probability distribution function}
  \acro{PET}[PET]{positron emission tomography}
  \acro{PSF}[PSF]{point spread function}
  \acro{RKHS}[RKHS]{reproducing kernel Hilbert spaces}
  \acro{SGD}[SGD]{stochastic gradient descent}
  \acro{SIM}[SIM]{structured-illumination microscopy}
  \acro{SNR}[SNR]{signal-to-noise ratio}
  \acro{SPECT}[SPECT]{single-photon emission computed tomography}
  \acro{SSIM}[SSIM]{structural similarity index}
  \acro{TCIA}[TCIA]{The Cancer Imaging Archive}
  \acro{TV}[TV]{total variation}
  \acro{USC-SIPI}[USC-SIPI]{University of Southern California Signal and Image Processing Institute}
\end{acronym}

\chapter{Introduction}
Biomedical imaging is a vast and diverse field:
there are a plethora of imaging devices using, e.g.,
light, X-rays, sound waves, magnetic fields, electrons, or protons,
to measure structures ranging from nano- to macroscale.
In many cases, computer software is needed to
turn the signals collected by the hardware
into a meaningful image.
These computer algorithms are
similarly diverse and numerous.
For example, Google Scholar lists around 1,600 papers
with the words ``MRI reconstruction'' in the title.

In this tutorial, we aim to present
a wide swath of biomedical image reconstruction algorithms
under a single framework---%
using a unified notation across domains,
modeling many modalities with a few basic operators,
and emphasizing commonalities among reconstruction algorithms.
Our goal is not to review the totality of the biomedical image reconstruction literature,
nor even to provide a comprehensive tutorial
(though we have included references throughout
that we hope will be useful pointers for a reader wishing to learn more).
Instead, we have focused on creating a brief and coherent narrative
that traverses some six decades of research.
The result is naturally a little stylized,
and reflects the biases of the authors.
For another recent perspective on these topics, see \cite{ravishankar_image_2020}.

Focusing on the similarities across modalities
is not only a didactic choice,
it is also a practical one.
While developing high-quality reconstruction algorithms
requires a deep understanding of both the physics of the imaging system
and the biomedical questions at hand,
we have found great benefit in our own work from adopting a unified point of view.
We model imaging modalities as combinations of a small set of building blocks,
which allows us to see connections between modalities
and quickly port ideas and computer code from one to the next.
Reconstruction algorithms can treat the imaging model as a black box,
meaning that one algorithm can work for many modalities.
Such modular solutions may not be completely optimal,
but they are often sufficient for a given application.
When they are not, they still serve as a valuable starting point.

The structure of the tutorial is as follows.
We begin by introducing the concept of forward models
and a set of building blocks for creating them in Section~\ref{chap:forward}.
In the next three sections
we introduce reconstruction algorithms in a roughly chronological way,
beginning with direct reconstruction and $\ell_2$ regularization (Section~\ref{chap:classical}),
followed by modern, sparsity-based techniques (Section~\ref{chap:modern}),
and ending with learning-based methods (Section~\ref{chap:learning}).
We compare these approaches and give our outlook on
the next steps for the field in Section~\ref{chap:conclusion}.

\chapter{Forward Models}
\label{chap:forward}
The goal of any imaging system is recover an image,
which is a map of a meaningful physical quantity
that varies over space.
In biomedical imaging, this might be
the concentration of fluorophores in a cell
or
the density of the tissues in a person's body.
Mathematically, we describe an image as a function, $f: \mathbb{R}^d \to \mathbb{R}$, %
that maps a point in space $\iv{x} \in \mathbb{R}^d$ to a real number, $y = f(\iv{x}) \in \mathbb{R}$.
We leave the dimension, $d$, of $\iv{x}$ unspecified;
it is usually between one and four: up to three spatial dimensions possibly plus time.

\begin{figure}[!htbp]
  \centering
  \includegraphics[width=\textwidth]{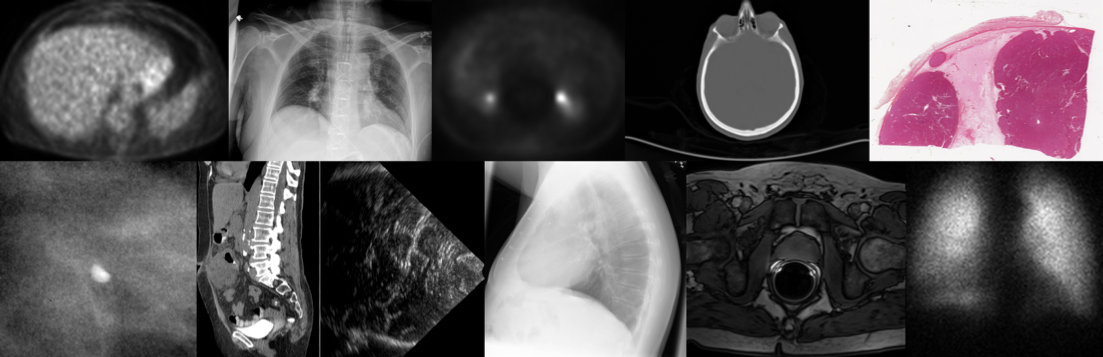}
  \caption{Examples of biomedical images.
    Top row:
    \acf{PET},
    computed radiography,
    X-ray \acf{CT} plus \ac{PET},
    X-ray \ac{CT},
    and microscopy.
    Bottom row:
    mammography,
    and contrast-enhanced X-ray \ac{CT},
    ultrasound,
    digital radiography,
    \acf{MRI},
    and \ac{PET} plus \acf{SPECT}.
    (All images from the \acf{TCIA}~\cite{clark_cancer_2013}
    except the microscopy~\cite{gutman_digital_2017}
    and ultrasound \cite{mercier_online_2012}.)}
  \label{fig:biomed-images}
\end{figure}

In any practical setting, the image of interest is not directly observed,
but must be reconstructed from measurements.
We view $f$ through a physical imaging system, $\mathcal{H}$,
that maps $f$ to a vector of measurements, 
\begin{equation}
\vv{g} = \mathcal{H}\{f\} \in \mathbb{R}^M.
\end{equation}
The measurement process encompasses all of the physics of the imaging system,
including the sources of noise.
It is nondeterministic, nonlinear, and impossible to perfectly characterize or simulate.
The measurements, on the other hand, are necessarily finite, quantized, and discrete
because they are stored digitally.%
\footnote{Throughout the text, we use bold font for these finite-length vectors and matrices
  to emphasize that they are concrete objects that we can store and use in computations.}
Biomedical imaging includes a broad diversity
of imaging modalities, e.g., systems that rely on
electromagnetic radiation (visible light, X-rays, gamma rays, etc.),
electrons, magnetic fields, or sound waves;
see Figure~\ref{fig:biomed-images} for examples.

Most%
\footnote{but not all, as we will discuss in Section~\ref{chap:learning}}
methods to recover an image from its measurements
rely on modeling the imaging system, $\mathcal{H}$,
with a deterministic operator $H$ and a random noise term, $n \in \mathbb{R}^M$,
so that
\begin{equation}
\mathcal{H}\{f\} \approx H\{f\} + n.
\end{equation}
Let us first focus on the deterministic part, $ H\{f\}$.
The key challenge is to develop an operator $H$ that is simultaneously mathematical convenient,
computationally tractable,
and physically accurate is a key challenge;
fortunately,
it turns out that many biomedical imaging modalities can be approximated
using combinations of just a few basic linear operators.
In this section, we will introduce these operators along with examples.

\section{Vector Spaces}

Before we can describe the action of an operator, $H$,
we need to specify its domain and range,
i.e., where it maps from and to.
Collectively, these are known as vector spaces,
which are the central topic of linear algebra.
For our purposes, it suffices to put the mathematical details aside
and define two special vector spaces that we will use throughout.

The first is $\mathbb{C}^N$, which is the space of finite-length lists of $N$ complex numbers.%
\footnote{In an engineering context, the term $\emph{vector}$
  usually refers to an element of $\mathbb{C}^N$ or $\mathbb{R}^N$;
  in the following, we have always tried to specify vector spaces explicitly to avoid ambiguity.}
Generally, we use elements of this space to represent measurements
or coefficients.
The use of complex-valued vectors is mathematically convenient,
but can be a source of frustration during implementation---%
on one hand, some programming languages do not natively handle complex values,
one the other, some papers specify algorithms only for real values.
When trouble arises, it is useful to remember that
a vector of $N$ complex numbers can be represented by a vector of $2N$ real numbers,
provided all arithmetic performed on the vector respects the convention used.

The second is the space of square-integrable functions $L_2(\mathbb{R}^d) = \{f : \int_{\mathbb{R}^d} |f(\iv{x})|^2 d\iv{x} < \infty \}$.
It is also called the space of finite-energy functions
and is closely related to the concept of energy in physics.
We use elements of this space to represent continuous objects,
including the image we are trying to reconstruct.
While the choice of $L_2(\mathbb{R}^d)$ cannot always be physically justified,
it is mathematically convenient and hardly restrictive on the resulting reconstructions.
Other restrictions,
(e.g., on the size of the non-zero support of a function,
its bandwidth,
or its smoothness)
may also be important depending on the specifics of the imaging system;
we will mention these as they arise.
We can generalize the concept of $L_2(\mathbb{R}^d)$ to $L_p(\mathbb{R}^d)$
in a straightforward way: $L_p(\mathbb{R}^d) = \{f : \int_{\mathbb{R}^d} |f(\iv{x})|^p d\iv{x} < \infty \}$,
with  $L_1(\mathbb{R}^d)$ being the most common.

\section{Linear Operators}
Our focus here is on operators that are linear.
\begin{definition}
  Let $\mathcal{X}$ and $\mathcal{Y}$ be two vector spaces.
  An operator, $H : \mathcal{X} \rightarrow \mathcal{Y}$,
  is a \emph{linear operator from $\mathcal{X}$ to $\mathcal{Y}$}
  when it satisfies
  \begin{equation}
    H\{af + bg \} = aH\{f\} + bH\{g\},
  \end{equation}
  for all scalars $a$ and $b \in \mathbb{R}$ and images $f$ and $g \in \mathcal{X}$.
\end{definition}
Linear operators provide a good approximation for many real imaging systems.
At the same time, they are mathematically convenient,
e.g., they allow the reconstruction problem to be solved efficiently
using convex optimization (Section~\ref{sec:iterative-reconstruction}).  
Another implication of selecting $H$ to be linear is that,
when $\mathcal{X}$ is $L_2(\mathbb{R}^d)$ and $\mathcal{Y}$ is $\mathbb{C}^M$,
each measurement can be written as the inner product between the image and a detector function
\begin{equation}
  \label{eq:detector}
  [H\{f\}]_m = [\vv{g}]_m = \ip{\eta_m}{f} = \int_{\mathbb{R}^d} \eta_m(\iv{x}) f(\iv{x}) d\iv{x},
\end{equation} 
where $[\vv{g}]_m$ designates the $m$th element of a vector $\vv{g}$.%
\footnote{This result holds for certain vector spaces other than $L_2(\mathbb{R}^d)$ as well;
we give a special case here for the sake of simplicity.}
The specific form of each $\eta_m$ is determined by the system model, $H$.
This result follows from a well-known theorem in analysis
called the Riesz representation theorem~\cite{rudin_real_1986,brezis_functional_2010}.

\section{Building Blocks}
The composition of linear operators is also linear,
which means that it is useful to build up a toolbox of
simple, well-understood operators from which to build more complicated ones.
We begin with sampling.
\begin{definition}
  Let $\mathcal{X}$ be a vector space of continuous functions.
  The \emph{sampling operator},
  $\samp_X : \mathcal{X} \rightarrow  \mathbb{R}^M$,
  with $ X = {\{\iv{x}_m\}_{m=1}^M}$,
returns a vector of the values of $f$ at $M$ known locations,
\begin{equation}
  [\samp_X\{f\}]_m = f(\iv{x}_m).
\end{equation}
See Figure~\ref{fig:sampling} for an example.
\end{definition}

\begin{figure}[htbp]
  \definecolor{flagBlue}{HTML}{003F87}
  \definecolor{flagYellow}{HTML}{FCD856}
  \definecolor{flagGreen}{HTML}{007A3D}
  \centering

  \newsavebox\mybox
  \begin{lrbox}{\mybox}
    \begin{tikzpicture}
      \draw [->] (0,0) -- (.39\textwidth,0) node [below] {$x_1$};
      \draw [->] (0,0) -- (0,.19\textwidth) node [left] {$x_2$};
      \fill[gray] (.5,1.5) circle (3pt) node[black, below] {$\iv{x}_1$};
      \fill [gray] (2.5,2) circle (3pt) node[black, below] {$\iv{x}_2$};
      \fill [gray] (4.3,.6) circle (3pt) node[black, below] {$\iv{x}_3$};
    \end{tikzpicture}
  \end{lrbox}

  \begin{tikzpicture}
    
    \node [inner sep=0, label=below:$f$] (input)
    {\includegraphics[width=.38\textwidth]{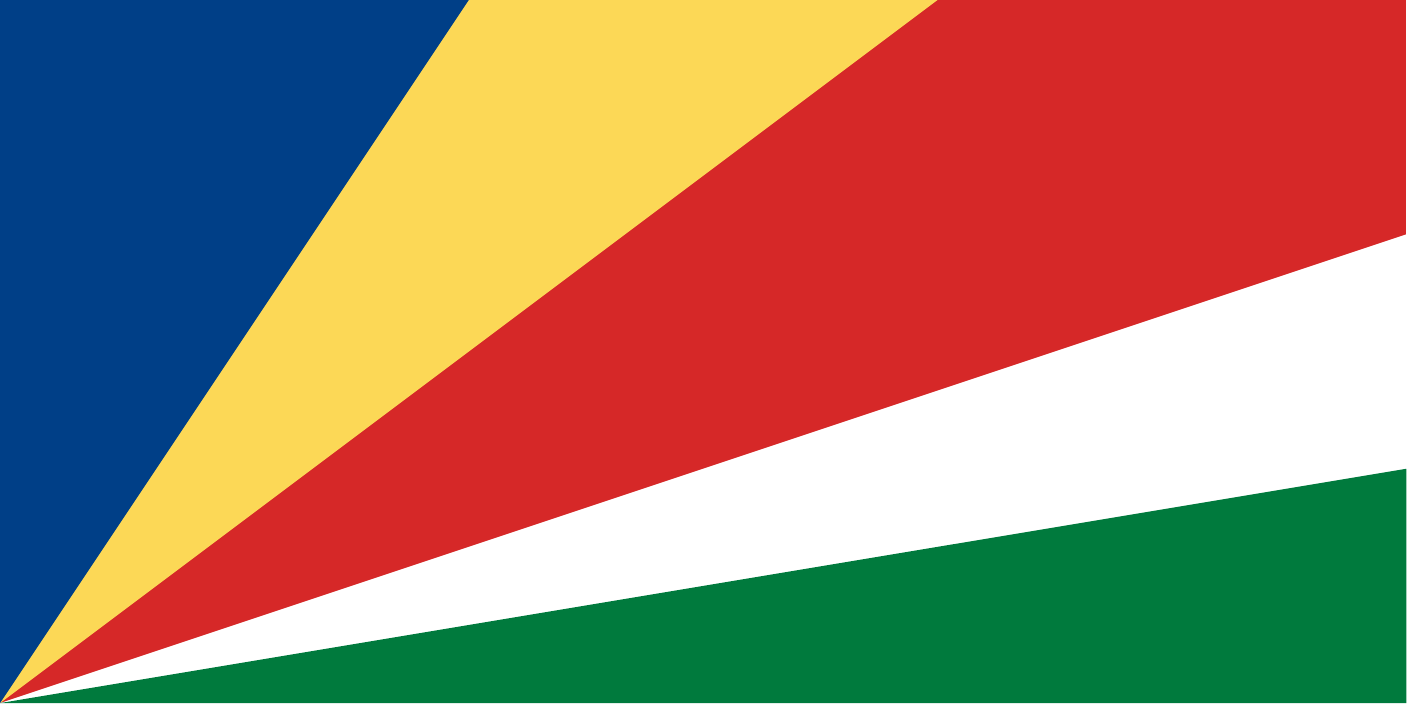}};

    \draw [->] (input.south west) -- (input.south east) -- ++(.25,0) node [below] {$x_1$};
    \draw [->] (input.south west) -- (input.north west) -- ++(0,.25) node [left] {$x_2$};

    \node[rectangle, fill=gray!20, right = 1.5 of input] (S) {$\samp_X$};
    
    \node[below = 1.5 of S, inner sep=0] (x) {\usebox\mybox};
    \node[below = -.5 of x] {$\{\iv{x}_m\}$};

    \node (output) [right = 1.5 of S, inner sep=0] {$
      \begin{bmatrix}
        \color{flagBlue} \blacksquare \\
        \color{flagYellow} \blacksquare \\
        \color{flagGreen} \blacksquare 
      \end{bmatrix}$
    };
    \node[below = 0 of output] {$\vv{g}$};
    
    \draw[->,line width=1mm] (input) -- (S);
    \draw[->,line width=1mm] (S) -- (output);
    \draw[->,line width=1mm] (x) -- (S);
  \end{tikzpicture}

  \caption{Sampling the flag of Seychelles, $f$,
    at the three locations specified by $\{\iv{x}_m\}$
    results in the vector of measurements $\vv{g}$.}
  \label{fig:sampling}
\end{figure}
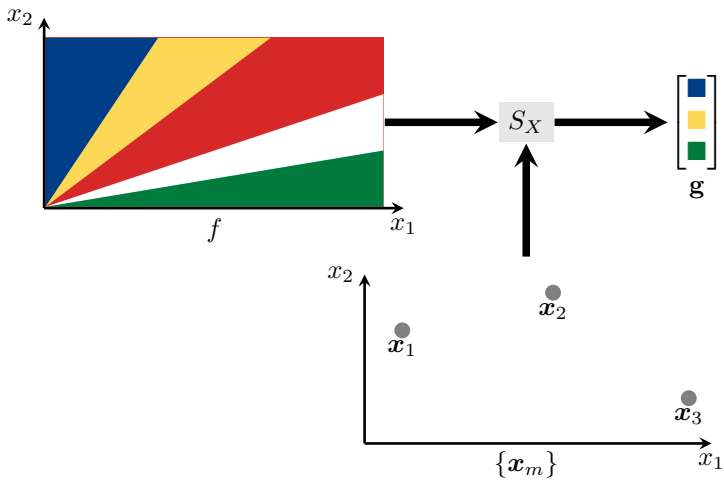
For the sampling operator to be well-defined on $\mathcal{X}$,
the elements of the space need to be continuous functions,
This is not the case for $ L_2(\mathbb{R}^d)$,
however, many subspaces of $ L_2(\mathbb{R}^d)$ work,
e.g., the space of bandlimited functions.
A Hilbert space over which sampling can be defined is called 
\iac{RKHS},
see \cite{manton_primer_2015} for more details.

Sampling is not a particularly useful imaging model on its own;
it is useful because it allows us to express imaging as a sampled version of a well-understood continuous operation,
such as the Fourier transform.
\begin{definition}
The \emph{Fourier transform}, $\fourier :  L_2(\mathbb{R}^d) \cap L_1(\mathbb{R}^d) \rightarrow  L_2(\mathbb{R}^d)$, 
expresses $f$ in terms of its frequency components,
\begin{equation}
  \label{eq:fourier}
  \fourier\{f\}(\iv{\omega}) = \hat{f}(\iv{\omega}) = \int_{\mathbb{R}^d} f(\iv{x}) \e^{-\imath \ip{\iv{\omega}}{\iv{x}}} d\iv{x}.
\end{equation}
Here, we have introduced the common shorthand that
$\hat{f}$ stands  for the Fourier transform of $f$.
See Figure~\ref{fig:FT} for an example.
\end{definition}
The restriction to $f \in  L_2(\mathbb{R}^d) \cap L_1(\mathbb{R}^d)$ has two advantages:
First, it ensures that the integral \eqref{eq:fourier} is well defined.
Second, it guarantees that $\hat{f}$ is continuous,
which then makes sampling feasible.

\begin{figure}[htbp]

  \centering

  \begin{tikzpicture}
     \node [inner sep=0, label=below:$f$] (input)
    {\includegraphics[width=.38\textwidth]{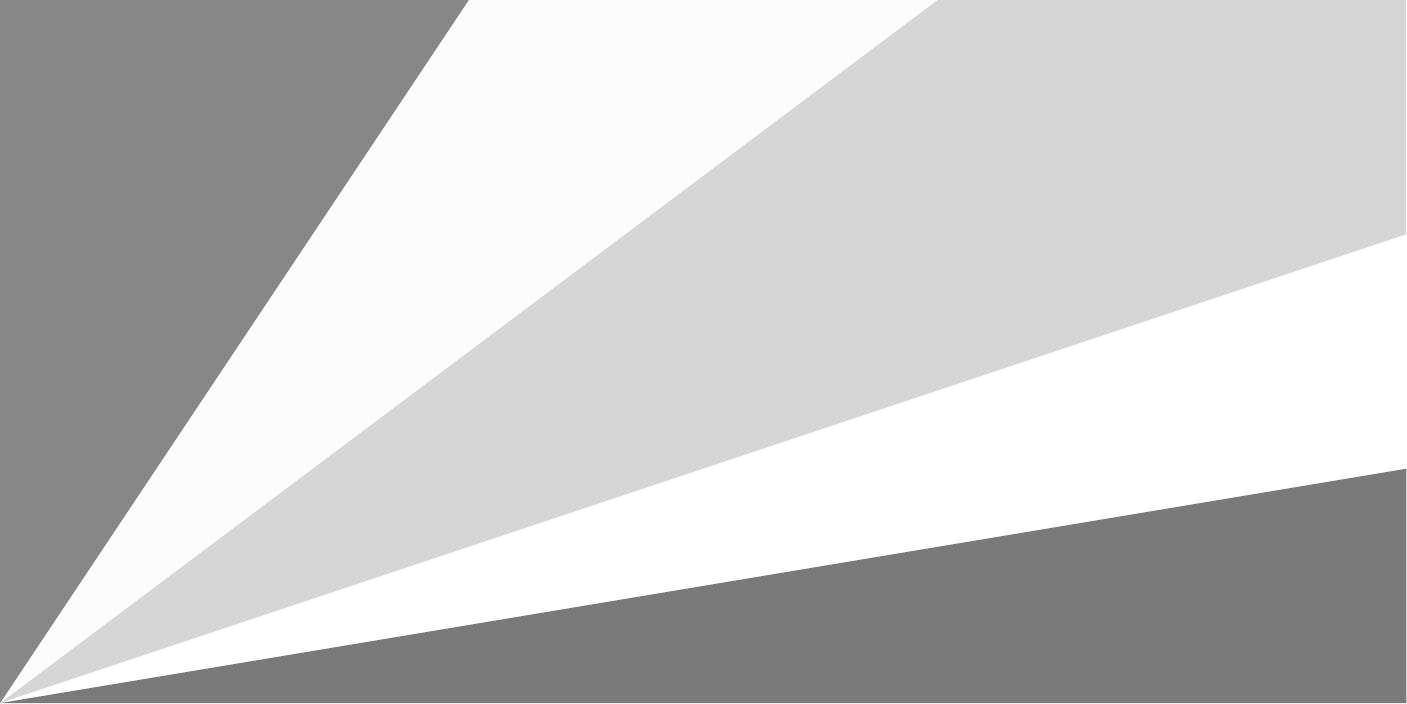}};
    \draw [->] (input.south west) -- (input.south east) -- ++(.25,0) node [below] {$x_1$};
    \draw [->] (input.south west) -- (input.north west) -- ++(0,.25) node [left] {$x_2$};
    
    \node[rectangle, fill=gray!20, right = 1 of input] (F) {$\fourier$};

    \node [inner sep=5, label=below:$\log |\hat{f}|$, right = 1 of F] (output)
    {\includegraphics[width=.25\textwidth]{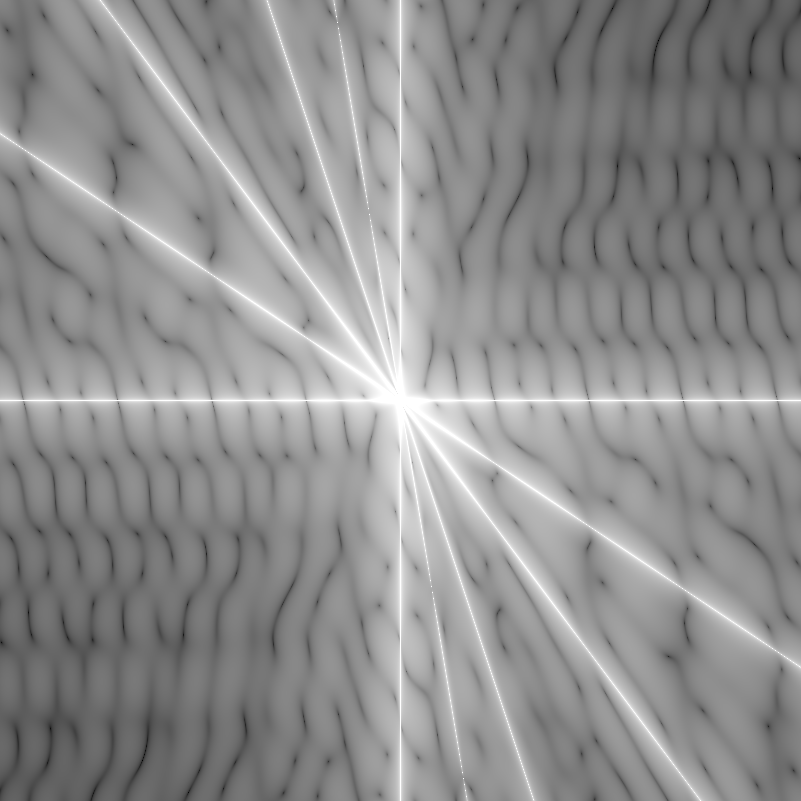}};

    \draw [<->] (output.west) -- (output.east) node [right] {$\omega_1$};
    \draw [<->] (output.south) -- (output.north) node [above] {$\omega_2$};

    \draw[->,line width=1mm] (input) -- (F);
    \draw[->,line width=1mm] (F) -- (output);
  \end{tikzpicture}

  \caption{Applying the Fourier transform to a grayscale rendition of the flag of Seychelles, $f$.
    Plotted here (in $\log$-scale) is the magnitude of the result,
    which is a complex-valued function of the frequency vector $\iv{\omega}$.
    High values (lighter pixels) in the Fourier transform correspond to
    the directions of greatest change;
    this explains the white lines in $\hat{f}$ perpendicular to the edges in the flag.
    The Fourier transform has been computed analytically thanks to \cite{lee_fourier_1983}.
}
  \label{fig:FT}
\end{figure}
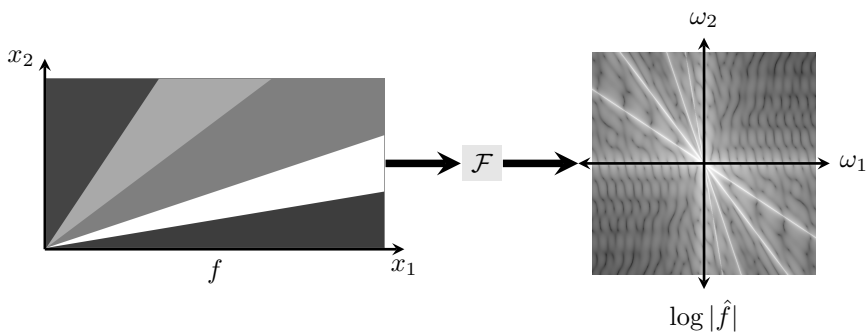

\begin{example}[Magnetic resonance imaging]
The Fourier transform and sampling operators are sufficient to model \ac{MRI};
this is known as the $k$-space formalism~\cite{webb_introduction_2002}.
In proton density-weighted \ac{MRI},
the physical quantity of interest is the proton density.
The measured signal is given by
\begin{equation}
\vv{g} = (\samp_X \circ \fourier) \{f\},
\end{equation}
where $\circ$ indicates operator composition,i.e., $(A \circ B)\{f\} = A\{ B \{ f\} \}$.
The locations of the Fourier samples depend on the hardware setup.
Typical choices include Cartesian sampling, random sampling, and spiral sampling.
\end{example}

\begin{definition}
The \emph{inverse Fourier transform}, $\fourier^{-1} :  L_2(\mathbb{R}^d) \rightarrow  L_2(\mathbb{R}^d)$, 
recovers $f$ from its Fourier transform,
\begin{equation}
\label{eq:inverse_fourier}
  f(\iv{x}) = \fourier^{-1} \{\hat{f}\}(\iv{x}) = \frac{1}{(2\pi)^d}\int_{\mathbb{R}^d} \hat{f}(\iv{\omega}) \e^{\imath \ip{\iv{\omega}}{\iv{x}}} d\iv{\omega}.
\end{equation}
\end{definition}

The fact that the Fourier transform is invertible suggests a method for MRI reconstruction: 
approximate the integral~\eqref{eq:inverse_fourier} with a sum over the known $[\vv{g}]_m = \hat{f}(\iv{\omega}_m)$.
As the measurement density increases, this method becomes more and more accurate.
We will explore this approach in more detail in Section~\ref{sec:direct}. 

\begin{definition}
The \emph{multiplication operator}, $\mult_h :  L_2(\mathbb{R}^d) \rightarrow L_2(\mathbb{R}^d) $,
multiplies $f$ pointwise by another image,
\begin{equation}
  \mult_h \{f \}(\iv{x}) = f(\iv{x}) h(\iv{x}),
\end{equation}
where $h$ must be bounded: for all $\iv{x}$, $|h(\iv{x})| < \infty$.
\end{definition}

One example of multiplication is windowing (Figure~\ref{fig:windowing}),
where $h$ is a function taking values between zero and one, such as a $\rect$.
\begin{equation}
\label{eq:rect}
  \rect(\iv{x}) =
  \begin{cases}
    1 & \iv{x} \in [-\frac{1}{2}, \frac{1}{2}]^d; \\
    0 & \text{otherwise}.
  \end{cases}
\end{equation}
While it is mathematically convenient
to work with infinite signals
(e.g., note the integral over all space in the definition of the inverse Fourier transform),
real imaging systems have a limited spatial extent;
windowing with $\rect$ conveniently models this limit.

\begin{figure}[htbp]
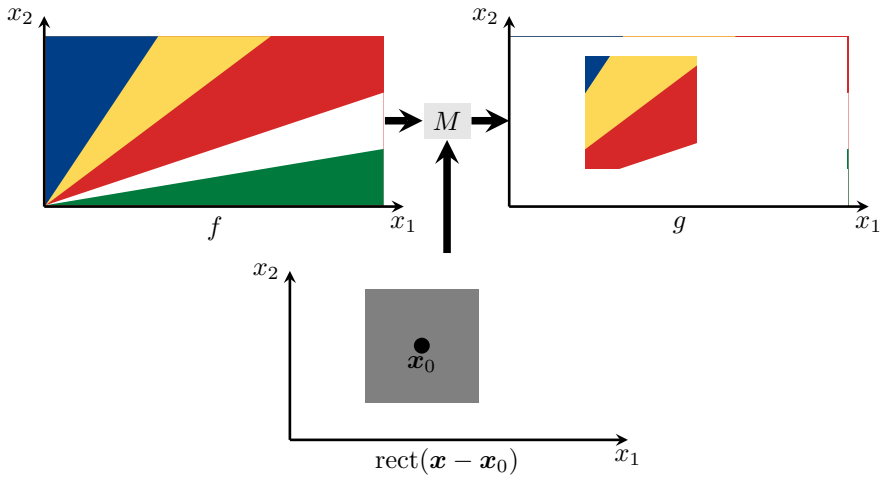

  \centering
  \newsavebox\myWindow
  \begin{lrbox}{\myWindow}
    \begin{tikzpicture}
      \draw [->] (0,0) -- (.38\textwidth,0) node [below] {$x_1$};
      \draw [->] (0,0) -- (0,.19\textwidth) node [left] {$x_2$};
      \fill[gray] (1, .5) rectangle (2.5, 2);
      \fill[black] (1.75,1.25) circle (3pt) node[below] {$\iv{x}_0$};
    \end{tikzpicture}
  \end{lrbox}

  \begin{tikzpicture}
    
    \node [inner sep=0, label=below:$f$] (input)
    {\includegraphics[width=.38\textwidth]{Figures/flag.pdf}};

    \draw [->] (input.south west) -- (input.south east) -- ++(.25,0) node [below] {$x_1$};
    \draw [->] (input.south west) -- (input.north west) -- ++(0,.25) node [left] {$x_2$};

    \node[rectangle, fill=gray!20, right = .5 of input] (S) {$\mult_h$};
    
    \node[below = 1.5 of S, inner sep=0] (x) {\usebox\myWindow};
    \node[below = -.5 of x] {$h(\iv{x}) = \rect (\iv{x} - \iv{x}_0)$};

     \node [right = .5 of S, inner sep=0, label=below:$g$] (output)
     {\includegraphics[width=.38\textwidth]{Figures/flag.pdf}};

     \fill[white] (output.south west) rectangle  ($ (output.south west) + (1,.19\textwidth) $);
     \fill[white] (output.south west) rectangle  ($ (output.south west) + (.38\textwidth,.5) $);
     \fill[white] ($(output.south west) + (1, 2) $) rectangle  ($ (output.south west) + (.38\textwidth,.19\textwidth) $);
     \fill[white] ($(output.south west) + (2.5, .5) $) rectangle  ($ (output.south west) + (.38\textwidth,.19\textwidth) $);
     
     \draw [->] (output.south west) -- (output.south east) -- ++(.25,0) node [below] {$x_1$};
     \draw [->] (output.south west) -- (output.north west) -- ++(0,.25) node [left] {$x_2$};

    \draw[->,line width=1mm] (input) -- (S);
    \draw[->,line width=1mm] (S) -- (output);
    \draw[->,line width=1mm] (x) -- (S);
  \end{tikzpicture}

  \caption{Windowing the flag of Seychelles, $f$,
    by multiplication with a shifted $\rect$.
    We can imagine this operation as modeling a limited field of view,
    like looking at the flag through a microscope.}
  \label{fig:windowing}
\end{figure}

Another example of multiplication  is modulation,
where $h$ is a complex exponential,
\begin{equation}
  h(\iv{x}) = \e^{\imath \ip{\iv{\omega}_0}{\iv{x}}}.
\end{equation}
The effect of modulation is easier to see in the frequency domain,
\begin{equation}
  (\fourier \circ \mult_{ \e^{\imath \ip{\iv{\omega}_0}{\iv{x}}}} ) \{ f \} (\iv{\omega}) =  \fourier\{f\}(\iv{\omega} - \iv{\omega}_0).
\end{equation}
Thus, multiplication by a complex exponential results in a shift in the frequency content of the image.

\begin{example}[MRI with coil sensitivity]
  A more accurate model of MRI accounts for spatial variations in the coil sensitivity,
  which can be modeled by a of $f$ prior to Fourier sampling,
\begin{equation}
  \vv{g} = (\samp_X \circ \fourier \circ \mult_h) \{f\},
\end{equation}
where $h$ is measured as part of calibration.
For a detailed treatment of the MRI forward model, see \cite{fessler_model_2010}.
\end{example}

\begin{definition}
The \emph{convolution-by-$h$ operator}, $\conv_h : L_2(\mathbb{R}^d) \to L_2(\mathbb{R}^d) $,
implements the linear, shift-invariant system with impulse response $h$,
\begin{equation}
  \conv_h \{f\} (\iv{x}) = \int_{\mathbb{R}^d} h(\iv{x} - \iv{x}') f(\iv{x}') d\iv{x}'.
\end{equation}
The impulse response, $h$, must be chosen so that the integral always exists;
e.g., it is sufficient for $h$ to be absolutely integrable
($ \int_{\mathbb{R}^d} |h(\iv{x})| d\iv{x} < \infty$). %
Convolution can be equivalently defined in the frequency domain,
\begin{equation}
  \conv_h \{f \}(\iv{x}) = (\fourier^{-1} \circ \mult_{\hat{h}} \circ \fourier) \{ f \} (\iv{x}),
\end{equation}
where $\hat{h}$ is the frequency response of the system, $\hat{h} = \fourier \{h\}$.
\end{definition}
Although  convolution is a composition of Fourier transforms and a multiplication,
it deserves its own symbol because of its ubiquitous role in imaging.
What makes it so useful is that many measurement systems are essentially shift-invariant,
because the underlying physics work the same everywhere and the sensors are placed on a regular grid,
e.g., pixels on a \ac{CCD} camera.
We say \emph{essentially} because a shift-invariant model
implicitly assumes that our measurements extend infinitely in all directions---%
while this is not the case, it can usually be handled properly with the help of boundary conditions.
Convolution is also important because discrete convolutions
can be implemented using the \ac{FFT},
which can result in fast reconstruction algorithms.

\begin{example}[Brightfield microscopy]
\label{ex:microscopy}
In brightfield microscopy,
what appears at the eyepiece can be modeled as
a blurred version of what is on the specimen plane,
\begin{equation}
  g = \conv_h \{ f \}.
\end{equation}
In two dimensions, the blurring function is called the Airy pattern,
\begin{equation}
  h(\iv{x}) =  \left(
    \frac{
      2 J_1 ( r )
    }{
      r
    } 
  \right)^2,
  \quad
  r = \frac{
    2 \pi a \|\iv{x}\|_2 }{
    \lambda R },
\end{equation} 
where $J$ is the Bessel function of the first kind, $a$ is the radius of the aperture,
$\lambda$ is the wavelength of the illumination,
and $R$ is the focal length.
The width of the central lobe of this function
plays an important role in the achievable resolution%
---the minimum distance that two point sources must be separated
 for them to be distinguished.
For visible light, this limit is around 250 nm.

For another perspective of the resolution limit,
we note that the Fourier transform of the Airy pattern
is a scaled version of the indicator function of a circle,
i.e.,
\begin{equation}
  \hat{h}(\iv{w}) =
  \begin{cases}
    C, &  \| \iv{\omega} \|_2 \le w_\text{max};\\
    0, & \text{otherwise}.
  \end{cases}
\end{equation}
Functions for which the Fourier transform
is zero outside of some radius, $w_\text{max}$, are called low-pass functions.
Writing the measurement process in the Fourier domain we have
\begin{align}
  \fourier g &= ( \mult_{\hat{h}} \circ \fourier) \{ f \} ,
\end{align}
and therefore the high-frequency content of $f$ has no effect on $g$
because everything outside of $w_\text{max}$
is multiplied by zero during the measurement.
For a textbook treatment of the pattern disk,
see, e.g., \cite{born_principles_2002};
for excellent introductory information on microscopy,
see~\cite{davidson_microscopyu_2001}.

\end{example}

\begin{example}[Structured-illumination microscopy]
  
The idea behind \ac{SIM} is to use a special illumination pattern to modulate the image $f$,
rearranging its frequency content so that its high frequencies are no longer destroyed by the measurement process.
The continuous forward model is
\begin{equation}
  g = (\conv_h \circ \mult_{p}) \{ f \},
\end{equation}
where $p$ is a complex exponential, $p(\iv{x}) = \e^{\imath \ip{\iv{\omega}_0}{\iv{x}}}$.
Viewed in the Fourier domain,
\begin{equation}
  \fourier \{g\} (\iv{\omega}) = ( \mult_{\hat{h}} \circ \fourier \circ \mult_{p}) \{ f \} (\iv{\omega}) \\
  = \hat{h}(\iv{\omega}) \fourier \{ f \} (\iv{\omega} - \iv{\omega}_0).
\end{equation}
To recover both the high-frequency and low-frequency content of $f$,
the imaging process can be repeated with different illumination patterns,
resulting in different values for the frequency shift, $\iv{\omega}_0$.
For more information on \ac{SIM}, see \cite{diaspro_super_2016}
\end{example}

\begin{definition}
  A \emph{change of variables operator}, $\COV_\varphi : L_2(\mathcal{X}) \rightarrow L_2(\mathcal{Y})$,
  is a linear operator such that $\COV_\varphi f = f(\varphi(\cdot))$,
  where $\varphi:\mathcal{Y} \rightarrow \mathcal{X}$ is a diffeomorphism
(a smooth function with a smooth inverse).
\end{definition}
We have already seen one change of variables: the frequency shift caused by multiplication by a complex exponential.
Using a change of variables operator, we can express a shift as
\begin{equation}
  f(\iv{x} - \iv{x}_0) = \COV_{\varphi_{\iv{x}_0}} \{f\}(\iv{x}), \quad \text{with} \quad \varphi_{\iv{x}_0}(\iv{x}) = \iv{x} - \iv{x}_0.
\end{equation}
See Figure~\ref{fig:COV} for another example.

\begin{figure}[htbp]
  \definecolor{flagBlue}{HTML}{003F87}
  \definecolor{flagYellow}{HTML}{FCD856}
  \definecolor{flagRed}{HTML}{D62828}
  \definecolor{flagGreen}{HTML}{007A3D}
  \centering
  
 \newsavebox\myPolar
  \begin{lrbox}{\myPolar}
    \begin{tikzpicture}
      \begin{axis}
        [width=.5\textwidth,
        xtick={0, 1.57079},
        xticklabels={$0$, $\frac{\pi}{2}$},
        xlabel near ticks,
        xlabel=$\theta$,
        xmax = 2,
        ymin = 0,
        ymajorticks=false,
        ylabel=$r$,
        axis lines = middle,
        every axis y label/.style={
          at={(ticklabel* cs:1)},
          anchor=east,},
        hide obscured x ticks=false]
        
        \addplot[domain=0:0.1651,samples=50,smooth,fill=flagGreen,draw=none] {2 / cos(deg(x))}\closedcycle;
        \addplot[domain=0.1651:.321750,samples=50,smooth,fill=white,draw=none] {2 / cos(deg(x))}\closedcycle;
        \addplot[domain=0.321750:0.463647,samples=50,smooth,fill=flagRed, draw=none] {2 / cos(deg(x))}\closedcycle;
        \addplot[domain=0.463647:0.6435,samples=50,smooth,fill=flagRed, draw=none] {1 / sin(deg(x))}\closedcycle;
        \addplot[domain=0.6435:0.98279,samples=50,smooth,fill=flagYellow,draw=none] {1 / sin(deg(x))}\closedcycle;
        \addplot[name path=blue, domain=0.98279:pi/2,samples=50,smooth,fill=flagBlue,draw=none] {1 / sin(deg(x))} \closedcycle;
       
     \end{axis}
    \end{tikzpicture}
  \end{lrbox}
  
  \begin{tikzpicture}
    
    \node [inner sep=0, label=below:$f$] (input)
    {\includegraphics[width=.38\textwidth]{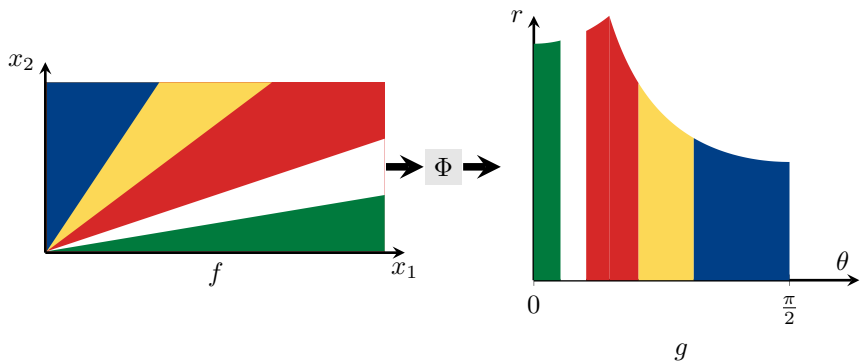}};

    \draw [->] (input.south west) -- (input.south east) -- ++(.25,0) node [below] {$x_1$};
    \draw [->] (input.south west) -- (input.north west) -- ++(0,.25) node [left] {$x_2$};

    \node[rectangle, fill=gray!20, right = .5 of input] (S) {$\COV_\varphi$};

     \node [right = .5 of S, inner sep=0, label=below:$g$] (output) {\usebox\myPolar};

    \draw[->,line width=1mm] (input) -- (S);
    \draw[->,line width=1mm] (S) -- (output);
  \end{tikzpicture}

  \caption{A change of variables applied to the flag of Seychelles, $f$.
    Each point is expressed in terms of its polar coordinates, i.e.,
    $\varphi(\theta, r) = (r \cos(\theta), r \sin(\theta))$.
    To make $\varphi$ a diffeomorphism, its domain must be restricted to
  $\theta \in [0, 2\pi)$ and $r \in (0, \infty]$.} %
  \label{fig:COV}
\end{figure}

\begin{example}[X-ray computed tomography]
  The idealized forward model for a variety of \ac{CT} imagining modalities,
  including X-ray \ac{CT}, \ac{PET}, \ac{SPECT}, and \ac{ET} is called the X-ray transform.
It measures the line integrals of an image and is usually given in 2D by
\begin{equation}
  \xray_\theta \{f\} (y) = \int_{\mathbb{R}} f(y \iv{\theta} + t \iv{\theta}^\perp) dt,
\end{equation}
where $\theta$ is an angle, $\iv{\theta} = (\cos \theta, \sin \theta)$ is a unit vector,
and $\iv{\theta}^\perp = (-\sin \theta, \cos \theta)$ is its perpendicular. 
It turns out that the X-ray transform can also be written in terms of the building blocks we have already defined as
\begin{equation}
  (\fourier_{\text{1D}} \circ \xray_\theta) \{f\}(\omega) = (\COV_{\varphi_\theta} \circ \fourier_{\text{2D}}) \{f\}(\omega), 
\end{equation}
where $\varphi(\omega) = \omega \iv{\theta}$
and the subscripts on the Fourier transforms remind us that
that the image is two dimensional while its X-ray transform is one dimensional.
We verify this by writing
\begin{align}
  (\COV_{\varphi_\theta}  \circ \fourier_{\text{2D}}) \{ f \}(\omega) 
  &=  \int_{\mathbb{R}^2} f(\iv{x}) \e^{-\imath \ip{ \omega \iv{\theta} }{\iv{x}}} d\iv{x}\\ %
  &=  \int_{\mathbb{R}^2} f(\vv{A}\vv{a}) \e^{-\imath \ip{\omega \vv{\theta}}{\vv{A}\vv{a}}} d\vv{a}\\ %
  &=  \int_{\mathbb{R}^2} f(\vv{A}\vv{a}) \e^{-\imath \omega [\vv{a}]_1} d\vv{a}\\ %
  &= \int_{\mathbb{R}}\int_{\mathbb{R}} f(\vv{A}\vv{a}) d[\vv{a}]_2  \e^{-\imath \omega [\vv{a}]_1} d[\vv{a}]_1\\ %
  &=  \fourier_{\text{1D}} \circ \xray_\theta f (\omega), %
\end{align}
where the trick is that we can always use the orthogonal matrix $\vv{A} =
\begin{bmatrix}
  \vv{\theta} & \vv{\theta}^\perp
\end{bmatrix}$ to express $\iv{x}$ as $\vv{A} \vv{a}$ for some $\vv{a}$.
The result is called the Fourier slice theorem or central slice theorem,
because it relates the 1D Fourier transform of the X-ray transform of an image
to a slice through its 2D Fourier transform.
For more information see \cite{kak_principles_2001} for X-ray \ac{CT},
with additional details on \ac{PET} in \cite{bailey_positron_2005}
and \ac{ET} in \cite{frank_electron_2006}.

\end{example}

\section{Discretization}
\label{sec:discretization}
We have, so far, presented a toolbox of operators that can be combined to model many biomedical imaging modalities.
Generally speaking, we will need to implement these operators in software
to solve image reconstruction problems.
But, how can we do this when the definitions are in the continuous domain?
That is, how can we represent $f \in L_2(\mathbb{R}^n)$ in a computer?
This is the problem of discretization.

One approach to the problem is to discretize \emph{at the end}:
formulate the reconstruction algorithm in the continuous domain and develop discrete versions of each of the necessary operations.
Anywhere a function is needed, it is represented by a finite vector of its samples;
the evaluation of the function at other points is achieved via interpolation, and
integrals are replaced with sums by the trapezoidal rule, etc.
We will discuss this approach further in Section~\ref{sec:direct}.

We focus here on discretizing \emph{at the beginning} via the synthesis operator.
\begin{definition}
  The \emph{synthesis operator} $\synth_{\{\beta_n\}_{n=1}^{N}} : \mathbb{R}^N \rightarrow L_2(\mathbb{R}^d) $,
  expresses a function as a sum of a finite number of basis functions,
  \begin{equation}
    f(\iv{x}) = \synth_{\{\beta_n\}_{n=1}^{N}} \{\vv{f} \}(\iv{x}) = \sum_{n=1}^{N} [\vv{f}]_n\beta_n(\iv{x}),
  \end{equation}
  where $\beta_n \in  L_2(\mathbb{R}^d)$,
  $f$ denotes the original function and
$\vv{f}$ is the vector of its coefficients.
\end{definition}
The synthesis operator allows for a simple discretization of
any continuous forward model, $H: L_2(\mathbb{R}^d) \rightarrow \mathbb{R}^M$,
according to $\vv{H} = H \circ \synth$, where $\vv{H}$ is an $M \times N$ matrix, called the system matrix. 
The discretized imaging model is then $\vv{g} = \vv{H} \vv{f}$.
This discretization approach is also called \emph{finite series expansion};
see \cite{censor_finite_1983} for an early discussion.
Note the dimensions of the system matrix:
$M$ is the number of measurements and 
$N$ is the number unknowns to reconstruct;
these can both easily be on the order of millions.
So, while $\vv{H}$ is a matrix, it is usually too large to store.
Instead, its elements are computed whenever they are needed.

Recalling that we can view the linear measurement process as taking inner products with detector functions \eqref{eq:detector},
we can express the elements of the system matrix as inner products between the basis functions and these detector functions,
\begin{equation}
   [\vv{H}]_{m,n} =  \ip{\eta_m}{\beta_n} = \int_{\mathbb{R}^d} \eta_m(\iv{x}) \beta_n(\iv{x}) d\iv{x}.
\end{equation}

For simplicity,
the basis functions are often chosen to be shift-invariant, $\beta_n(\iv{x}) = \beta(\iv{x} - \iv{x}_n)$,
where the $\iv{x}_n$ form a regular grid,
e.g., 
$\iv{x}_n \in \{\iv{x} = T \iv{k} + \iv{x}_0 : \iv{k} \in [0, 1, \dots, N_1-1] \times [0, 1, \dots, N_2-1]\}$.
They may also be multiplicatively separable,
\begin{equation}
  \beta(\iv{x}) = \beta(x_1)\beta(x_2)\cdots \beta(x_d),
\end{equation}
where (in a slight abuse of notation) $\beta(\iv{x})$ is the $d$-dimensional separable basis function and $\beta(x)$ is the 1D function from which it is built.

\begin{example}[Pixel model]
Because the ultimate result of image reconstruction is a digital image,
it is natural to use the square pixel as a basis function, i.e., $\beta(\iv{x}) = \rect(\iv{x})$,
see Figure~\ref{fig:pixel} for an example.
The downside of the pixel model is, generally, its lack of smoothness:
It is not differentiable, so it is not appropriate for continuous forward models involving differentiation.
It is also not bandlimited:
its Fourier transform is
\begin{equation}
\label{eq:sinc}
\fourier\{\rect\}(\omega) = \sinc(\omega) = \frac{\sin \pi \omega}{\pi \omega}
\end{equation}
and therefore may not interact well with forward models that amplify high frequencies.
Despite this, the pixel model is ubiquitous in practice.
\end{example}
\begin{figure}[!htbp]
  \centering
  \begin{subfigure}{.38\linewidth}
      \includegraphics[width=\textwidth]{Figures/flag.pdf}
    \caption{no discretization}
  \end{subfigure}
  \begin{subfigure}{.38\linewidth}
    \includegraphics[width=\textwidth]{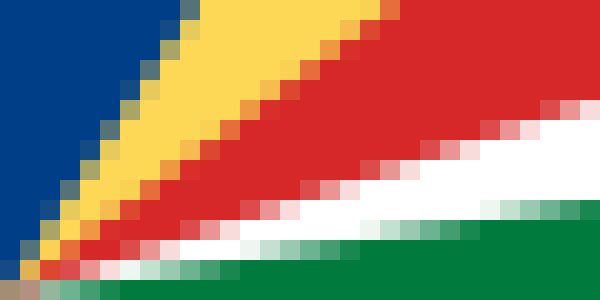}
    \caption{$30 \times 15$ pixels}
  \end{subfigure}
  \begin{subfigure}{.38\linewidth}
    \includegraphics[width=\textwidth]{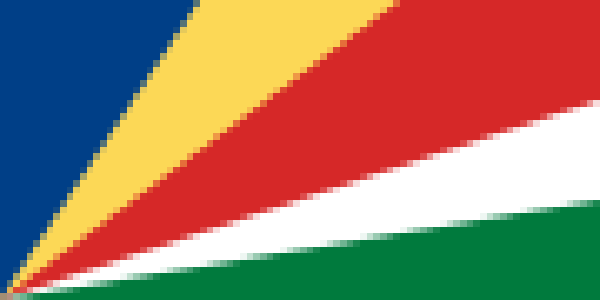}
    \caption{$90 \times 45$ pixels}
  \end{subfigure}
  \begin{subfigure}{.38\linewidth}
    \includegraphics[width=\textwidth]{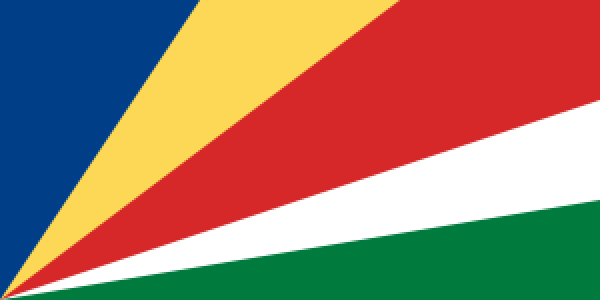}
    \caption{$300 \times 150$ pixels}
  \end{subfigure}
  \caption{The flag of Seychelles, discretized in three different pixel bases.
    Very fine discretizations, like (d), can be arbitrarily close to a target function,
    but, in general, they are not equal.
    In the digital version of this document,
    you should be able to zoom in on (d) and see the pixels,
    while (a) should remain sharp no matter the zoom level.
  }
  \label{fig:pixel}
\end{figure}

Although they are beyond the scope of the current discussion,
there are many more tools that can be used for discretization,
e.g., wavelets~\cite{mallat_wavelet_2009} and dictionary methods~\cite{tosic_dictionary_2011}.

\section{Summary}
We have presented a set of building blocks---%
linear operators that can be composed
to model a variety of physical measurement processes.
In the following, we give a nonexhaustive
list of these imaging modalities.
It is important to note that these examples are simplified:
state-of-the-art models typically account for additional physical effects.
One must also include the effect of the measurement noise,
which we discuss in the following sections.

\begin{description}[style=unboxed,leftmargin=0cm]
\item[2D or 3D tomography] uses X-rays.
  We can model the measurements as $\{ (\samp_{X_i} \circ \xray_{\theta_i}) \{f\} \}$,
  which is a set of sampled X-ray projections
  taken from different directions,
  where $i$ indexes the projection direction.
  Variations include parallel
  (where the sampling locations do not depend on the projection directions, $X_i = X$),
  cone beam, and spiral sampling patterns.

\item[3D microscopy] uses fluorescence.
  We can model the measurements as $ (\samp_{\mathbb{Z}^3} \circ \conv_h )\{ f \}$,
  which is a blurred, sampled version of $f$.
  Variations include brightfield, confocal, and light sheet microscopy.

\item[\ac{SIM}] uses fluorescence.
  We can model the measurements as $\{( \samp_{\mathbb{Z}^2} \circ \conv_h \circ \mult_{w_i}) \{ f \} \}_i$,
  which is a set of modulated, blurred, and sampled versions $f$,
  where $i$ indexes the modulation pattern.
  Variations include performing 3D reconstructions
  or using non-sinusoidal patterns.

\item[\ac{PET}] uses gamma rays.
  We can model the measurements as $\{( \samp_{X_i} \circ \xray_{\theta_i}) \{f\} \}$,
  which, like in tomography, are sampled X-ray projections.
  We use the subscript $X_i$ to indicate that the sampling positions
  are dependent on the projection direction,
  and note that there may be as few as one sample per direction.
  Variations include time-of-flight \ac{PET}.

\item[\ac{MRI}] uses magnetic fields and measures radio waves.
  We can model the measurements as  $\{ (\samp_X \circ \fourier \circ \mult_{w_i}) \{ f \} \}$,
  which are a set of samples of the Fourier transform of $f$
  after weighting by several different coil sensitivities,
  indexed by $i$.
  There are several variations in the sampling pattern,
  including Cartesian and radial sampling,
  and the method can also be used for dynamic processes,
  e.g., cardiac \ac{MRI}.

\item[Optical diffraction tomography] uses coherent light.
  Although the problem is, in general, nonlinear,
  using the first Born approximation~\cite{wolf_three_1969},
  we can model the measurements as $ (\samp_X \circ \fourier  \circ  \mult_h ) \{ f \} $,
  where $h$ is related to the incident wave.
  Variations include holography and grating-based interferometry.

\end{description}

In each case above, stated a measurement model
that goes from a continuously defined $f: \mathbb{R}^d \to \mathbb{R}$
to a discrete $\vv{g} \in \mathbb{C}^M$.
By additionally making use of a synthesis operator,
each of these can be converted to
a system matrix that relates
the discrete image we want to recover
to the discrete measurements we have access to.
This matrix plays a key role in reversing the measurement process
to arrive at an estimate of the image of interest,
which is the topic of the following sections.

\subsection{Further Reading}
Many of these building blocks are fundamental concepts in signal processing;
for a textbook treatment, see the classic \cite{oppenheim_discrete_2009}
or \cite{vetterli_foundations_2014}.
For more information on the modalities we have discussed,
see the references listed at the end of each example.
For a MATLAB software library based on this methodology,
see \cite{unser_globalbioim_2017}.

\chapter{Classical Image Reconstruction}
\label{chap:classical}
Assuming that we have an accurate, implementable
forward model, $H$, and measurements $\vv{g}$,
our task now is to reconstruct an image.
Our focus is on discretizing immediately;
i.e., using a synthesis operator to create a system matrix, $\vv{H}$,
because doing so gives access to simple, implementable
reconstruction algorithms.
But, before dispensing with the operator $H$
in favor of the matrix $\vv{H}$,
we need to discuss direct inversion methods,
which are derived in the continuum.

\section{Direct Inversion}
\label{sec:direct}
Many of the imaging models we have discussed are invertible
under the assumption that there is no noise and
a continuum of measurements:
there is a formula that relates the measurements to the image
we are trying to recover.
As a concrete example,
we recall that the Fourier transform is invertible.
Thus, 
since \ac{MRI} is based on the Fourier transform,
it seems that the inverse Fourier transform should provide
an \ac{MRI} reconstruction technique.
To develop a direct inversion algorithm,
we formulate a forward model in the continuum (i.e., ignoring sampling),
$H : L_2(\mathbb{R}^d) \rightarrow L_2(\mathbb{R}^d) $
and derive (analytically) its inverse, $H^{-1}$.
Such an inverse will usually involve integrals of a continuous version
of the measurements.
To actually implement  $H^{-1}$,
we replace these integrals with sums
over the known measurements
and use interpolation whenever we need to evaluate a quantity
in between its samples.
See \cite{lewitt_reconstruction_1983} for a early discussion
of this approach.

\begin{example}[Filtered back projection]
  The canonical example of direct inversion is the \ac{FBP} algorithm
  for X-ray \ac{CT}, which dates back to 1971~\cite{ramachandran_three_1971}.
  Starting with the X-ray transform as a forward model,
  we use the Fourier slice theorem to derive the following inverse
  \begin{equation}
    \label{eq:FBP}
    H^{-1} \{g\} (\iv{x}) = \int_0^{\pi}
    \fourier^{-1}_{\text{1D}} \mult_{|\omega|}  \fourier_{\text{1D}}
    \xray_\theta f \left( \langle \vv{\theta}^\perp,\iv{x}\rangle \right) d\theta,
  \end{equation}
  with $\iv{x} \in \mathbb{R}^2$ and $\iv{\theta}^\perp = (-\sin \theta, \cos \theta)$.
  This equation tells us to take each projection (from 0 to $\pi$),
  convolve by a filter with impulse response $|\omega|$,
  and, finally, back project it (smear it across the image).
  The name filtered back projection is therefore a bit misleading
  because
  the filtering occurs before the back projection.

  The formula \eqref{eq:FBP} looks simple, but several approximations
  are needed to implement it from discrete measurements.
  First, the integral over $\theta$ must be replaced with a sum.
  Second, the filter must be implemented digitally,
  with careful handling of the boundary, high frequencies
  (because they are amplified),
  and zero frequency (because it is removed).
  Finally, the filtered projections must be interpolated
  to back project them.
  Despite these approximations, the \ac{FBP} is a robust
  and high-quality method,
  especially when the number of projections is large.
  
\end{example}

\section{Variational Methods}
To avoid the challenges associated
with discretizing a direct inversion method,
we prefer to work with a linear forward model
that takes the form of a matrix, $\vv{H}$.
In this setting,
the simplest approach to image reconstruction would be to solve 
\begin{equation}
\label{eq:problem_0}
\vv{H}\vv{f} = \vv{g}
\end{equation}
for $\vv{f}$.%
\footnote{While the components of $\vv{f}$ are actually
  expansion coefficients in some underlying basis
  (see Section~\ref{sec:discretization}),
  we shall, with a slight abuse of language,
  refer to them as pixels.}
The first problem with this formulation
is that $\vv{H}$ is rarely invertible---%
in most scenarios, we want to get as many image pixels, $N$,
out of as few measurements, $M$, as possible,
leading to a short, wide system matrix with no left inverse.
Practically, this means that there will be many $\vv{f}$'s that satisfy \eqref{eq:problem_0}.
If any of the measurement vectors are linearly dependent
(which can happen in any setting,
but is guaranteed when $M > N$),
any amount of measurement noise will be enough to move $\vv{g}$
out of the span of $\vv{H}$, leading to \eqref{eq:problem_0} having no solution.
Finally, even when a unique solution does exist ($M=N$),
the solution can be very sensitive to noise,
to the point were direct inversion of $\vv{H}$ is useless
(see Figure~\ref{fig:conv_example} for an example).
Thus, \eqref{eq:problem_0} is an example of an ill-posed problem.

\begin{definition}
A \emph{well-posed} problem satisfies:
\begin{enumerate}
\item A solution exists.
\item The solution is unique.
\item The solution is a continuous function of the measurements.
\end{enumerate}
A problem that is not well-posed is call \emph{ill-posed}.
\end{definition}

\begin{figure}[!htpb]
  \centering
  \begin{subfigure}[t]{.3\linewidth}
    \includegraphics[width=\linewidth]{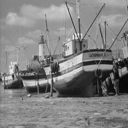}
    \caption{image, $\vv{f}$}
  \end{subfigure}
  \begin{subfigure}[t]{.3\linewidth}
    \includegraphics[width=\linewidth]{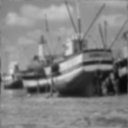}
    \caption{measurements, $\vv{H} \vv{f} = \vv{g}$}
  \end{subfigure} 
   \begin{subfigure}[t]{.3\linewidth}
    \includegraphics[width=\linewidth]{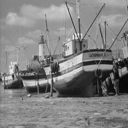}
    \caption{direct inversion, $\vv{H}^{-1}\vv{g}$}
  \end{subfigure}\\
  \begin{subfigure}[t]{.3\linewidth}
    \includegraphics[width=\linewidth]{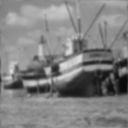}
    \caption{noisy measurement, $\vv{g} +\vv{n}$}
  \end{subfigure}
  \begin{subfigure}[t]{.3\linewidth}
    \includegraphics[width=\linewidth]{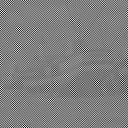}
    \caption{noisy direct inversion, $ \vv{H}^{-1} (\vv{g} +\vv{n})$}
  \end{subfigure}
  \caption{Effect of measurement noise on direct inversion of discrete convolution.
    Each pixel of the measurement (b) is a local average of pixels in the image (a).
    This operation is invertible (c), but even a tiny amount of noise (d)
    is amplified by the inversion, leading to heavy artifacts in the reconstruction (e).
  Image from the \acf{USC-SIPI} Image Database.}
  \label{fig:conv_example}
\end{figure}

A more promising approach,
which is the focus of this section and the next,
is to formulate reconstruction as a minimization problem.
These approaches are sometimes called \emph{variational methods};
while this term has precise meanings in fields such as
differential equations, quantum mechanics, and Bayesian inference,
in our context it only means we aim to minimize something.
As a first example,
we can seek a reconstruction that minimizes the disagreement with the measurements according to
\begin{equation}
  \tilde{\vv{f}} = \argmin_{\vv{f} \in \mathbb{R}^N} \| \vv{g} - \vv{H} \vv{f} \|^2_2.
  \label{eq:problem_LS}
\end{equation}
From a geometric perspective, we see that this problem must have
at least one solution:
we can project $\vv{g}$ onto the span of $\vv{H}$,
resulting in a denoised measurement, $\vv{g}'$
and solve $\vv{H}{\vv{f}} = \vv{g}'$.
However, the problem may still be ill-posed.
Equating the gradient
(i.e., the vector of partial derivatives with respect to the unknown
$\vv{f} = (f_1, \dots, f_N)$)
of the objective function in \eqref{eq:problem_LS}
with zero results in the so-called normal equation,
\begin{equation}
  \vv{H}^* \vv{g} = \vv{H}^* \vv{H} \vv{f},
  \label{eq:normal}
\end{equation}
which $\vv{f}$ must satisfy to be a solution.
Here, $\vv{H}^*$ denotes the conjugate transpose.
When the number of linearly-independent measurements (rank of $\vv{H}$)
is at least the number of unknowns (length of $\vv{f}$),
then there is a unique solution $ (\vv{H}^* \vv{H})^{-1}\vv{H}^* \vv{g} =\vv{f}$.
Otherwise, there are infinitely-many solutions.
Specifically, the solutions form an affine subspace of $\mathbb{R}^N$,
$\mathcal{N}(\vv{H}) + \tilde{\vv{f}}$, where $\mathcal{N}(\vv{H}) = \{\vv{f} : \vv{H}\vv{f} = \vv{0} \}$ is called the null space of $\vv{H}$ and $\tilde{\vv{f}}$ is one solution of \eqref{eq:problem_LS}.

\begin{remark}[Does the square matter?] \label{rem:square}
  \begin{shaded}
      One challenge in understanding the biomedical reconstruction literature
      is the plethora of similar-looking problem statements that
      may or may not be equivalent.
      The first such example is 
      \begin{equation}
        \argmin_\vv{f} \| \vv{g} - \vv{H} \vv{f} \|_2 
        \quad \text{versus} \quad 
        \argmin_\vv{f} \| \vv{g} - \vv{H} \vv{f} \|^2_2.
      \end{equation}
      These expressions are equivalent because the function $|\cdot|^2 : x \mapsto |x|^2$ is monotonically increasing
      over the nonnegative real numbers.
      We will see a less-trivial generalization of this idea in Remark~\ref{rem:squareTik}.
      
      The version without the square is more natural:
      it gives a reconstruction that agrees with the measurements in the
      sense of Euclidean distance.
      The benefit of the square is that it cancels the square root
      in the norm, which leads to a simpler expression for the gradient.
  \end{shaded}
    \end{remark}

    Unfortunately, the infinitely-many solutions case is more the norm
    than the exception.
This is, again, because we want the highest possible resolution, $N$,
from a fixed number of measurements, $M$, determined by the hardware.
As a result, the number of unknowns is often larger than the number of measurements.
Thus, \eqref{eq:problem_LS} is generally an ill-posed problem.

\begin{example}[Null space]
Consider a toy problem where
\begin{equation}
  \vv{H} =
  \begin{bmatrix}
    1 & 0 & 1\\
    0 & 1 & -1 \\
    1 & 1 & 0
  \end{bmatrix},
  \quad \text{and} \quad
  \vv{f}_{\text{true}} = 
  \begin{bmatrix}
    1 \\ 1\\ 2
  \end{bmatrix}.
\end{equation}
The true measurement process is $\vv{g} = \vv{H}\vv{f}_{\text{true}} + \vv{n} =
\begin{bmatrix}
  3 & -1 & 2.1
\end{bmatrix}^T
$ and we want to reconstruct by solving \eqref{eq:problem_LS}.
If we use a standard solver
(MATLAB's \texttt{pcg} function)
on the normal equation \eqref{eq:normal},
the result is $\tilde{\vv{f}} \approx
\begin{bmatrix}
    1.70 &
    0.37 &
    1.33
\end{bmatrix}^T$, giving a sum of squared errors of 0.0033.
But, by changing the initialization of the solver,
infinitely many equally good solutions can be obtained:
\begin{center}
\begin{tabular}{@{}lll@{}}
\toprule
  $\vv{f}_0$ & $\tilde{\vv{f}}$ & $\| \vv{g} - \vv{H} \vv{f} \|^2_2$ \\ \midrule
$\begin{bmatrix}
    0 & 0 & 0
\end{bmatrix}^T$ &  
$ \begin{bmatrix}
    1.70 & 0.37 & 1.33
\end{bmatrix}^T$ & .0033 \\
$\begin{bmatrix}
    0 & 0 & 1
\end{bmatrix}^T$ &  
$ \begin{bmatrix}
1.37&
    0.70&
    1.67
\end{bmatrix}^T$ & .0033 \\
$ \begin{bmatrix}
    13 & 8 & 18
\end{bmatrix}^T $&  
$ \begin{bmatrix}
-2.30 & 4.37 & 5.33
\end{bmatrix}^T$ & .0033 \\
\multicolumn{1}{c}{\vdots} & \multicolumn{1}{c}{\vdots} & \multicolumn{1}{c}{\vdots} \\
\bottomrule
\end{tabular}
\end{center} 
This effect occurs because $\vv{H}$ has a non-zero null space,
specifically $\vv{H} \begin{bmatrix} 1 & -1 &-1 \end{bmatrix}^T = \vv{0}$;
note that each solution can be obtained from another by adding a scaled version
of this null space vector.
\end{example}

\section{Tikhonov regularization}
When a formulation is ill-posed in the sense of having no solution,
it is a dead end.
But, when a formulation merely has infinitely-many solutions,
all we need is a mechanism to pick one of them to obtain a well-posed problem;
Tikhonov regularization provides a classical mechanism to do this.

The idea is to introduce a regularization functional, $ \| \vv{L}\vv{f} \|^2_2$, 
where $\vv{L}$ is a linear operator that measures the ``roughness'' of $\vv{f}$
(or some other undesirable property).
The reconstruction can then be formulated as
\begin{equation}
  \label{eq:problem_tikhonov_constrained}
  \argmin_\vv{f}  \| \vv{L}\vv{f} \|^2_2
  \quad \text{subject to} \quad
  \| \vv{g} - \vv{H} \vv{f} \|^2_2 \le \sigma^2,
\end{equation}
where $\sigma$ is a scalar constant selected based on the expected level of noise.
An equivalent formulation is 
\begin{equation}
  \label{eq:problem_tikhonov}
  \argmin_\vv{f}  \| \vv{g} - \vv{H} \vv{f} \|^2_2 + \lambda  \| \vv{L}\vv{f} \|^2_2
  \end{equation}
  where $\lambda$ is a scalar constant controlling the strength of the regularization (see Figure~\ref{fig:effect-of-lambda}).
  See Remark~\ref{rem:const-vs-reg} for a discussion of the equivalence
  of these two problems.

\begin{figure}[!htbp]
  \centering
  \hfill%
  \begin{subfigure}[b]{.2\linewidth}
    \includegraphics[width=\linewidth]{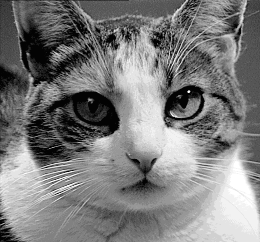}    
    \caption{ground truth image}
  \end{subfigure}\hfill
 \begin{subfigure}[b]{.2\linewidth}
    \includegraphics[width=\linewidth]{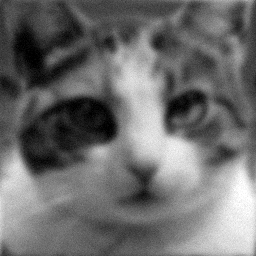}    
    \caption{degraded image}
  \end{subfigure}\hfill\null\\
 \begin{subfigure}[b]{\linewidth}
  \includegraphics[width=.2\linewidth]{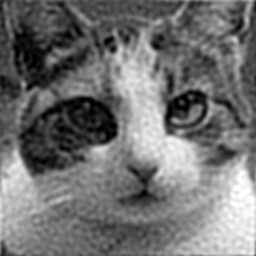}%
  \includegraphics[width=.2\linewidth]{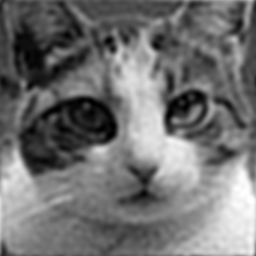}%
  \includegraphics[width=.2\linewidth]{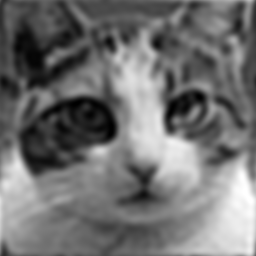}%
  \includegraphics[width=.2\linewidth]{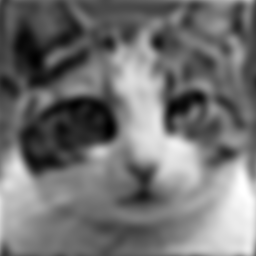}%
  \includegraphics[width=.2\linewidth]{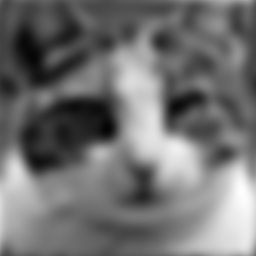}
    \caption{reconstructed images, increasing $\lambda \rightarrow$}
  \end{subfigure}\\
\caption{Illustration of the effect of changing the Tikhonov regularization parameter, $\lambda$,
  when the regularization term is the $\ell_2$-norm of the gradient,
  which favors smoothness.
Stronger regularization reduces noise, but also reduces sharpness;
this can be thought of as a kind of bias--variance tradeoff.}
\label{fig:effect-of-lambda}

\end{figure}

\begin{remark}[Does the square matter now?] \label{rem:squareTik}
  \begin{shaded}
    We saw in Remark~\ref{rem:square} that the square in 
    the least squares formulation is only a mathematical convenience.
    Now, with the addition of the regularization term,
    we can compare a more general case,
    \begin{equation}
      \label{eq:sq-now-vs}
      \argmin_\vv{f} J_1(\vv{f})  + \lambda_{\text{L}} J_2( \vv{f} )
      \quad \text{versus} \quad 
      \argmin_\vv{f} \Psi_1 ( J_1(\vv{f} ))  + \lambda_{\text{R}} \Psi_2 (J_2( \vv{f} )),
    \end{equation}
    where $J_1$ and $J_2$ are positive, differentiable functionals,
    $\lambda_{\text{L}}$ and $\lambda_{\text{R}}$ are positive numbers,
    and $\Psi_1$ and $\Psi_2$ are monotonically increasing, differentiable functions.
    Intuitively, these are different problems when $\lambda_{\text{L}} = \lambda_{\text{R}}$
    because the $\Psi$s will change the relative importance of $J_1$ and $J_2$
    (which we can think of as representing the data and regularization terms, respectively).
    The problems are equivalent in the sense that
    for any value of $\lambda_{\text{L}}$, we can always find a value of  $\lambda_{\text{R}}$ such that the problems are the same.
    To see this, take a solution of the left problem, $\tilde{\vv{f}}$.
    At this point, the gradients of $J_1$ and $\lambda_{\text{L}} J_2$
    are equal and opposite.
    In the right problem, at $\tilde{\vv{f}}$,
    the chain rule tells us that the gradient of the $J_1$ term
    has been multiplied by  $\Psi_1'( J_1( \tilde{\vv{f}} ))$
    and the gradient of the $J_2$ term
    has been multiplied  by  $\Psi_2'( J_2( \tilde{\vv{f}} ))$.
    Because both of these multipliers are positive,
    we can set
    \begin{equation}
    \lambda_{\text{R}} = \frac{\Psi_1'( J_1( \tilde{\vv{f}} )) \lambda_{\text{L}}}{\Psi_2'( J_2( \tilde{\vv{f}} ))},
  \end{equation}
    which makes $\tilde{\vv{f}}$ a critical point for the right problem in \eqref{eq:sq-now-vs}.
    We can make the same argument in the reverse direction,
    and a similar argument can be made with nondifferentiable functions by using subgradients.

    In practice, we set $\lambda$ by hand to give the best-looking reconstruction;
    this argument tells us that if we can find a good $\lambda$ for the left formulation,
    we will be able to find a good $\lambda$ in the right formulation.
    As a result, we have a lot of flexibility to add or remove, e.g., squares,
    from data fidelity and regularization terms without affecting the set of solutions we get.
  \end{shaded}
\end{remark}

In the Tikhonov formulation,
the normal equation becomes
\begin{equation}
  \vv{H}^* \vv{g} = (\vv{H}^* \vv{H} + \lambda \vv{L}^* \vv{L}) \vv{f}.
\end{equation}
This expression reveals the main advantage of the formulation:
the problem will have a unique solution when the intersection of the null spaces of $\vv{H}$ and $\vv{L}$ is zero.
One easy way to achieve this is to let $\vv{L}$ be full rank,
 e.g., when $\vv{L}$ is the identity matrix, 
the solution is unique and the regularization penalizes large values in $\vv{f}$.

\begin{remark}[Constraints versus regularization]
  \label{rem:const-vs-reg}
  \begin{shaded}
  We have now seen (and will continue to see) data and regularization terms
  on $\vv{f}$ expressed as both functionals (also called penalties), e.g., $\|\vv{L} \vv{f}\|_2$,
  and as constraints, e.g., $\|\vv{g} - \vv{H}\vv{f}\|_2 \le K$.
  The natural question is to compare
  \begin{equation}
    \argmin_{\vv{f}} J_1(\vv{f}) + \lambda J_2(\vv{f})
    \quad \text{versus} \quad
    \argmin_{\vv{f}} J_1(\vv{f}) \quad \text{s.t.} \quad J_2(\vv{f}) \le \sigma,
  \end{equation}
  where $J_1$ and $J_2$ are positive functionals and $\lambda$ and $\sigma$ are positive numbers.
  We will call these the penalized and constrained formulations;
  the first is also called Tikhonov regularization and the second
  Ivanov or Morozov regularization, depending on whether
  $J_1$ is the data term (Ivanov) or the regularization term (Morozov).
  We can afford to be loose in interchanging these forms
  because, roughly, solutions of one are also solutions of the other (when the $\lambda$ and $\sigma$ are chosen correctly).
  We make this notion precise in the next two theorems;
  for a discussion of the engineering implications, skip ahead.

  \begin{theorem}
    For every choice of $\lambda$, for each solution of the penalized problem,
    there is a choice of $\sigma$ that makes it a solution of the constrained problem. 
  \end{theorem}
  \begin{proof}
    Fix $\lambda$ and assume $\vv{f}^P$ is a solution to the penalized problem.
    Set $\sigma = J_2(\vv{f}^P)$.
    Now for sake of contradiction, assume there exists $\vv{f}^C$
    that is a better solution of the constrained problem than $\vv{f}^P$,
    i.e., $J_1(\vv{f}^C) < J_1(\vv{f}^P) $ and $J_2(\vv{f}^C) \le \sigma$.
    Then we have $J_1(\vv{f}^R) + \lambda J_2(\vv{f}^R) < J_1(\vv{f}^L) + \lambda J_2(\vv{f}^L)$,
    which contradicts $\vv{f}^P$ being a solution to the penalized problem.
    Thus, there is no such $\vv{f}^C$, which means $\vv{f}^P$
    is a solution to the constrained problem.
  \end{proof}

  \begin{theorem}
    Assume that the penalized problem has a solution
    (not necessarily unique) for every $\lambda \ge 0$.
    Then,
    for every choice of $\sigma$ and any each solution of the constrained problem,
    there is a choice of $\lambda$ that makes it a solution of the penalized problem.
  \end{theorem}
  \begin{proof}
    Fix $\sigma$ and assume $\vv{f}^C$ is a solution to the constrained problem.
    Consider the function $F(\lambda) = \min_{\vv{f}} J_1(\vv{f}) + \lambda J_2(\vv{f})$.
    We are searching for a value, $\lambda^*$, such that
    $F(\lambda^*) = J_1(\vv{f}^C) + \lambda^* J_2(\vv{f}^C)$, because this will make $\vv{f}^C$ a solution to the penalized problem.
    We know that, for any $\lambda^*$, $F(0) = \min_{\vv{f}} J_1(\vv{f}) \le J_1(\vv{f}^C) + \lambda^* J_2(\vv{f}^C)$
    and that we can pick $\lambda^\wedge$ large enough such that $F(\lambda^\wedge) \ge  J_1(\vv{f}^C) + \lambda^* J_2(\vv{f}^C)$.
    We can show that $F$ is continuous;
    thus, the intermediate value theorem gives us
    that there exists $\lambda^*$ such that $F(\lambda^*) = J_1(\vv{f}^C) + \lambda^* J_2(\vv{f}^C)$,
    which makes $\vv{f}^C$ a solution of the penalized problem.
  \end{proof}

  Practically, these problems are equivalent;
  theoretically, the relationship is more subtle
  (e.g., we have not shown that the solution sets can be made the same,
  only that the problems can be made to share a solution).
  When more is assumed about the functionals, we can say more.
  Generally, which form is chosen is a matter of personal preference.
  For a much deeper look at this issue, see~\cite{ciak_homogeneous_2012}.
  \end{shaded}
\end{remark}

\section{Bayesian Formulation}
\label{sec:classical-stat}

We now consider a statistical measurement model %
\begin{equation}
  \label{eq:model-stats}
  \vv{g} = \vv{H} \vv{f} + \vv{n},
\end{equation}
where $\vv{H}$ remains a deterministic measurement matrix,
but where $\vv{f}$ and $\vv{n}$ (and therefore $\vv{g}$) are random variables.
If we know (or can model or estimate)
the distribution of $\vv{n}$,
then the model \eqref{eq:model-stats} lets us determine
the \ac{PDF} of the measurements given a fixed image, $p( \vv{g} \mid \vv{f} )$.

To do reconstruction, we are more interested in
the distribution of $\vv{f}$ conditional on the measurements,
which we can obtain using Bayes' rule as
\begin{equation}
  p( \vv{f} \mid \vv{g} ) = \frac{p( \vv{g} \mid \vv{f} ) p(\vv{f})}{ p(\vv{g}) }.
   \label{eq:bayes}
\end{equation}
So, for a fixed set of measurements $\vv{g}$,
\eqref{eq:bayes} is a function 
that states how likely a given $\vv{f}$ is to have given rise to $\vv{g}$.
It requires a \emph{prior}---i.e., a model for $p(\vv{f})$---%
which is analogous to the choice of a regularizer
in the variational formulations we already presented.
Unfortunately, it also requires computation of $p(\vv{g})$,
which is typically intractable
(because it requires integrating over all $\vv{f}$'s that can create $\vv{g}$).
So, while \eqref{eq:bayes} in principle gives all the information
we might want in a reconstruction,
we usually can not use it directly.

Using the same statistical model,
we can build algorithms that seek a single image
satisfying some measure of optimality.
In one such approach, we seek the \ac{MMSE}
solution, which is given by
\begin{equation}
  \argmin_{\tilde{\vv{f}}}\mathbb{E}_{\vv{f}, \vv{g}}( \| \tilde{\vv{f}}(\vv{g}) - \vv{f}  \|^2_2)
  = \argmin_{\tilde{\vv{f}}} \int \| \tilde{\vv{f}}(\vv{g}) - \vv{f}  \|^2_2 p(\vv{f}, \vv{g})d\vv{f}d\vv{g},
  \label{eq:MMSE}
\end{equation}
where here $\tilde{\vv{f}}(\vv{g})$ is the reconstruction
(which depends on the measurements, $\vv{g}$)
and $\vv{f}$ is the ground truth.
One way to think of the \ac{MMSE} solution
is that if it is used to perform many reconstructions
from many different measurements,
it will have the lowest average error among all algorithms
(assuming that all the model assumptions are correct).
By equating the gradient of the functional in \eqref{eq:MMSE} with zero,
we find that the \ac{MMSE} is given by  $\tilde{\vv{f}} = \mathbb{E}(  \vv{f} \mid \vv{g}  )$,
which is the expectation of the conditional \ac{PDF}
given by \eqref{eq:bayes}.
Unfortunately, it is usually difficult to build algorithms
to find the \ac{MMSE} solution,
unless both $\vv{f}$ and $\vv{n}$ are
multivariate Gaussians.

A different approach is
to find the \ac{MAP} solution, which is
the mode of \eqref{eq:bayes},
\begin{equation}
  \label{eq:MAP}
  \argmax_{\vv{f}} p(\vv{f} \mid \vv{g} )
  = \argmax_{\vv{f}} p( \vv{g} \mid \vv{f} ) p(\vv{f}).
\end{equation}
Put another way, this is the $\vv{f}$ that is most likely
to have generated the measured $\vv{g}$.
As a word of caution:
when the prior on $\vv{f}$ and noise model are correct,
\ac{MAP} solutions do not, in general, follow them;
see \cite{nikolova_model_2007} for more discussion.
However, the \ac{MAP} approach is useful because
it provides a statistically coherent way
to account for noise and priors on the image,
while also leading to optimization problems that can be efficiently solved.

As a specific example of \ac{MMSE} and \ac{MAP} solutions,
we return to the measurement model \eqref{eq:model-stats} and
consider the case where the elements of the noise $\vv{n}$
are \ac{iid} Gaussian with zero mean and variance $\sigma^2$,
and $\vv{f}$ is a zero-mean Gaussian process
with invertible covariance matrix $\mathbb{E}(\vv{f}\vv{f}^*) = \vv{C}$.
In this case, we can transform the \ac{MAP} objective \eqref{eq:MAP}
using the negative log and arrive at
\begin{equation}
  \label{eq:gauss-MAP-objective}
  \argmin_{\vv{f}}\frac{1}{\sigma^2}  \| \vv{g} - \vv{H} \vv{f} \|^2_2 
  +  \| \vv{C}^{-\frac{1}{2}}\vv{f} \|^2_2,
\end{equation}
which has a closed-form solution
\begin{equation}
  \label{eq:gauss-MAP}
  (\vv{H}^* \vv{H} + \sigma^2 \vv{C}^{-1})^{-1} \vv{H}^* \vv{g}.
\end{equation}
One can also derive the corresponding \ac{MMSE} solution
by finding the expectation of \eqref{eq:bayes},
which essentially requires some lengthy algebra
involving the underlying multivariate Gaussian distributions.
The result is
\begin{equation}
  \label{eq:gauss-MMSE}
  \vv{C} \vv{H}^* \left( \vv{H} \vv{C} \vv{H}^* + \sigma^2 \vv{I} \right) ^{-1} \vv{g},
\end{equation}
which is also called the Wiener filter.
It turns out 
that \eqref{eq:gauss-MAP} and \eqref{eq:gauss-MMSE} are equal.
(One way to verify this is to check that both satisfy
the normal equation of \eqref{eq:gauss-MAP-objective},
$\vv{H}^* \vv{g} =  (\vv{H}^* \vv{H} + \sigma^2 \vv{C}^{-1}) \vv{f}$,
which is meaningful because the solution to \eqref{eq:gauss-MAP-objective} is unique.)
And, further, note that \eqref{eq:gauss-MAP-objective} is 
equivalent to the Tikhonov formulation \eqref{eq:problem_tikhonov}
with $\vv{L} = \vv{C}^{-\frac{1}{2}}$ and $\lambda = \sigma^2$;
thus, the optimal choice of $\vv{L}$ is the
whitening operator matched with $\vv{f}$.
So, for the special case of Gaussian denoising of a Gaussian process,
the \ac{MAP} and \ac{MMSE} solutions are the same,
and they can also be obtained from a specific variational formulation.
But, for many other image reconstruction problems,
this equivalence does not hold.
For a more in-depth discussion of the relationship of  \ac{MAP} and \ac{MMSE},
see \cite{gribonval_bayesian_2019}.

\section{Iterative Reconstruction}
\label{sec:iterative-reconstruction}
We have now seen several reconstruction problems that
have a closed-form solution.
Unfortunately, these usually involve inverting the normal matrix, $\vv{H}^* \vv{H}$,
which is large (its size is number of pixels squared).
So, for practical implementations, we turn to iterative algorithms.
Iterative algorithms for solving convex problems of this form are,
by now, a textbook topic~\cite{boyd_convex_2004};
we present one such algorithm here to give a sense of their form.

Given a Tikhonov reconstruction problem,
\begin{equation}
\argmin_\vv{f}  J(\vv{f}) = \argmin_\vv{f}  \| \vv{g} - \vv{H} \vv{f} \|^2_2 + \lambda  \| \vv{L}\vv{f} \|^2_2 .
\end{equation}
The gradient descent algorithm involves iterates of the form
\begin{equation}
  \label{eq:grad-desc}
  \vv{f}^{(k+1)} = \vv{f}^{(k)} - \gamma^{(k)} \nabla J(\vv{f}^{(k)}),
\end{equation}
where the gradient, $\nabla J(\vv{f})$, is given by 
\begin{equation}
    \label{eq:grad-desc-grad}
  \nabla J(\vv{f}) = - 2\vv{H}^*\vv{g} + 2(\vv{H}^* \vv{H} + \lambda\vv{L}^* \vv{L}) \vv{f}.
\end{equation}
Setting $\gamma$ to a small constant in \eqref{eq:grad-desc} gives a workable algorithm
that is also rather flexible.
For instance, convex constraints,
such as the positivity of $\vv{f}$,
can be enforced by projecting onto the feasible region after each iteration,
an algorithm known as projected gradient descent.
Gradient descent is good starting point for building iterative algorithms,
but there are many alternative.
The choice of the best algorithm
will generally depend on the specifics of problem at hand.
For example, the \ac{CG} algorithm solves the same problem
more efficiently, at the cost
of increased memory usage~\cite{shewchuk_introduction_1994}.

\begin{remark}[Convolutional normal matrix]
  \begin{shaded}
    It turns out that the normal matrix, $\vv{H}^* \vv{H}$,
    appears in many iterative reconstruction schemes
    and is often the computational bottleneck.
    The case of gradient descent~\eqref{eq:grad-desc-grad} is typical:
    reconstruction requires a single back projection, $\vv{H}^* \vv{g}$,
    followed by one normal operation $(\vv{H}^* \vv{H} + \lambda \vv{L}^* \vv{L}) \vv{f}$
    per iteration.
    Often, we select $\vv{L}$ to be a convolution (e.g., a high-pass filter),
    which means that  $\vv{L}^* \vv{L}$ is a convolution as well.
    If $\vv{H}^* \vv{H}$ is a convolution, then the entire operation is a convolution.
    This is important because a discrete, linear convolution can be computed using
    the \ac{FFT},
    resulting in a fast algorithm for the normal operation.
    Due to boundary conditions,
    the matrices may be Toeplitz (diagonal-constant) rather
    than circulant,
    but \ac{FFT}-based convolution can still be used for fast algorithms
    provided that proper zero-padding is applied;
    see \cite{oppenheim_discrete_2009} Section 8.7.

    For example, with a specific choice of discretization,
    it is possible to model  parallel-ray X-ray \ac{CT}
    with a convolutional normal matrix~\cite{delaney_fast_1996}.
    This approach has been used in
    cryo-EM~\cite{vonesch_fast_2011} and synchrotron microtomography~\cite{mccann_fast_2016}.
  \end{shaded}
  \end{remark}

\section{Summary}
Classical image reconstruction methods fall into either
direct methods or variational (regularized, iterative) methods.
Direct methods are derived in the continuous domain;
they are fast to apply and give good results when the number of measurements is high.
Variational methods involve minimizing an objective function
that usually consists of a data term and one or more regularization terms.
These methods tend to be more robust to noise than direct methods,
at the cost of increased computation.

\chapter{Sparsity-Based Image Reconstruction}
\label{chap:modern}
Over the past two decades,
numerous new reconstruction methods have emerged,
many of which continue to use the variational framework,
i.e., reconstruction is performed by minimizing a cost functional.
Of these, one dominant paradigm has been that of sparsity---%
the idea that a high dimensional image, $\vv{f}$,
can be represented by only a small number of nonzero coefficients.
The concept of sparsity is included in reconstruction problems
by replacing
the quadratic regularization functional, $\| \vv{L} \vv{f} \|^2_2$,
of Tikhonov regularization
to one that promotes sparsity.
Especially when the number or quality of the measurements is low,
such a change can drastically improve the reconstruction quality.
In this section, we will introduce these sparsity-based image reconstruction
techniques.

\section{Sparsity and Compressive Sensing}
As we saw in the previous section,
using classical reconstruction techniques,
recovering an image of $N$ pixels requires 
on the order of $N$ measurements.
The motivating idea behind compressive sensing
is that many images
can be compressed by applying a suitable linear transform 
and throwing away small values.
One way to model this compressibility
is to say that,
when expressed in the correct basis,
many images have few nonzero coefficients.
Mathematically, one can write that
\begin{equation}
  \| \vv{L} \vv{f} \|_0 = K \ll N.
\end{equation}
It turns out that, if $\vv{f}$ fits this model,
it can be recovered from a number of measurements
much smaller than its number of pixels.
Popular choices for the sparsifying transform include
finite differences, the Fourier transform, the discrete cosine transform, 
and the wavelet transform;
we compare a few of these in Figure~\ref{fig:transforms}.

\begin{figure}[!htbp]
  \centering
  \begin{subfigure}[b]{\linewidth}
    \begin{overpic}[width=.25\linewidth]{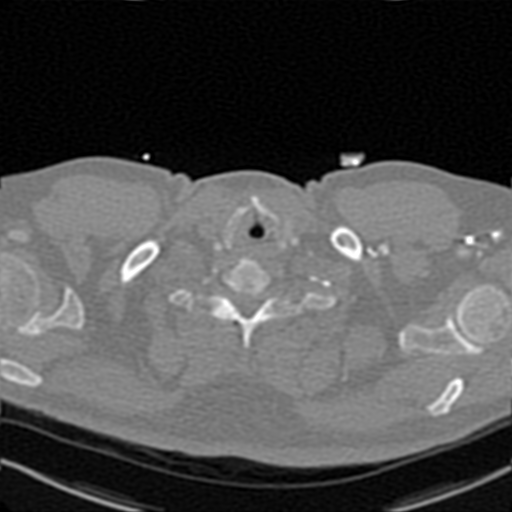}
      \put (2, 88) {\color{white} \footnotesize DFT, 26.9 dB}
    \end{overpic}%
    \begin{overpic}[width=.25\linewidth]{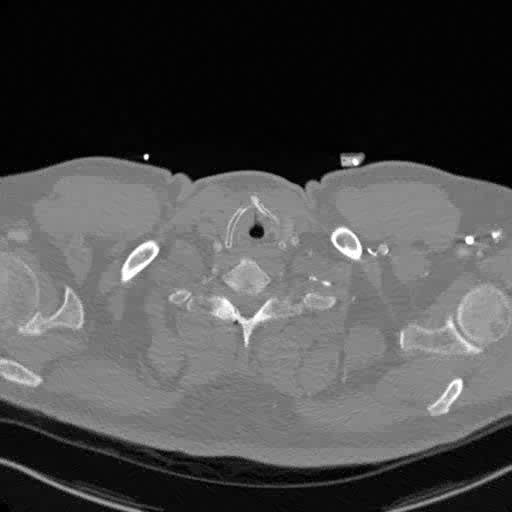}
      \put (2, 88) {\color{white} \footnotesize DCT, 28.8 dB}
    \end{overpic}%
   \begin{overpic}[width=.25\linewidth]{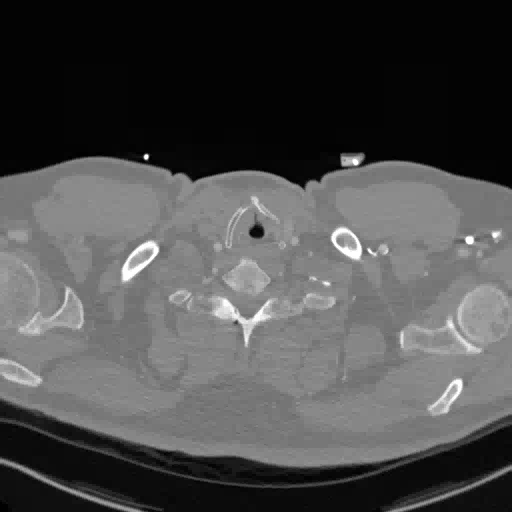}
      \put (2, 88) {\color{white} \footnotesize Haar, 28.8 dB}
    \end{overpic}%
   \begin{overpic}[width=.25\linewidth]{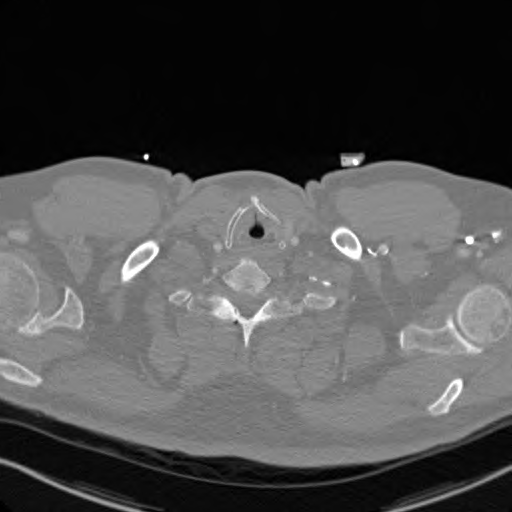}
      \put (2, 88) {\color{white} \footnotesize db4, 30.6 dB}
    \end{overpic}
\caption{Retaining 13,000 largest transform coefficients}
  \end{subfigure}
\par\bigskip %
  \begin{subfigure}[b]{\linewidth}
     \begin{overpic}[width=.25\linewidth]{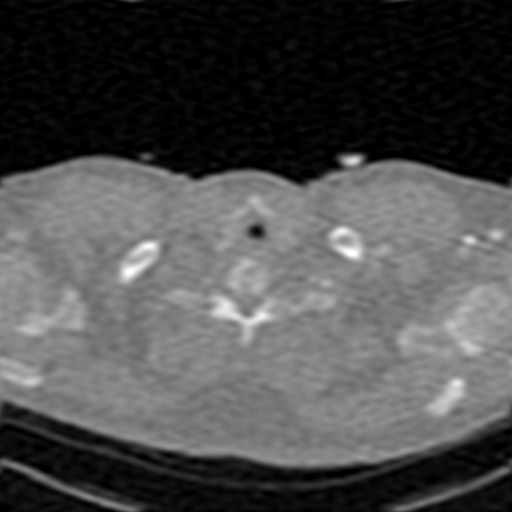}
      \put (2, 88) {\color{white} \footnotesize DFT, 20.6 dB}
    \end{overpic}%
    \begin{overpic}[width=.25\linewidth]{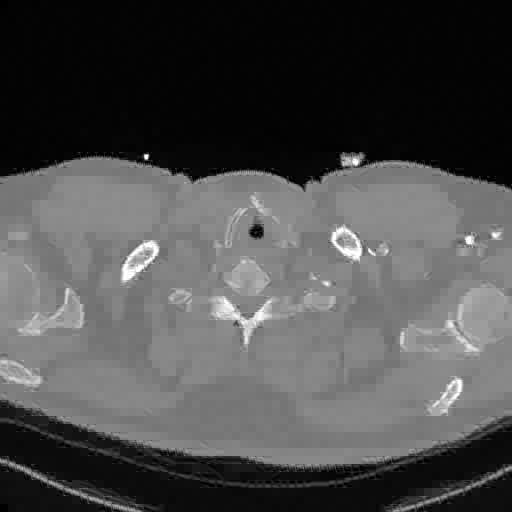}
      \put (2, 88) {\color{white} \footnotesize DCT, 23.8 dB}
    \end{overpic}%
   \begin{overpic}[width=.25\linewidth]{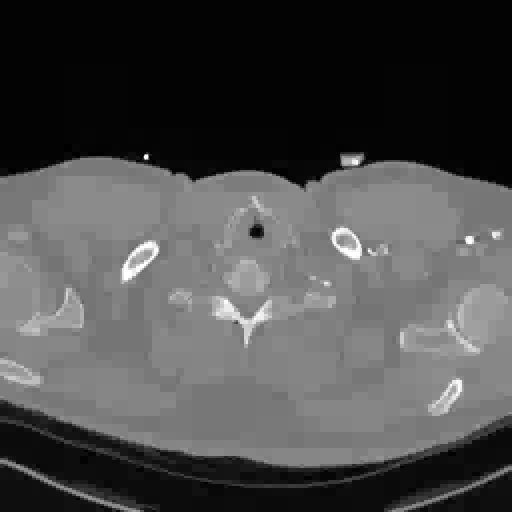}
      \put (2, 88) {\color{white} \footnotesize Haar, 23.1 dB}
    \end{overpic}%
   \begin{overpic}[width=.25\linewidth]{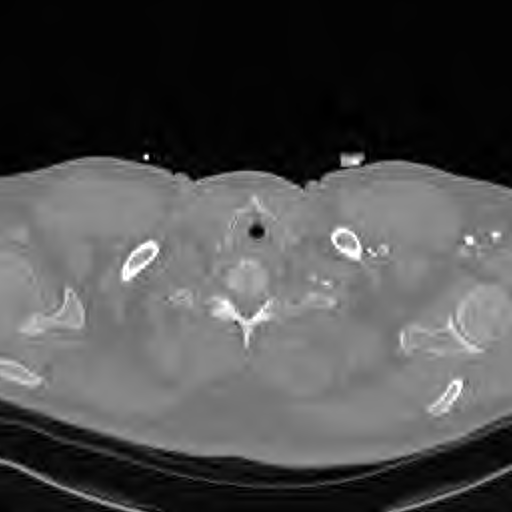}
      \put (2, 88) {\color{white} \footnotesize db4, 24.1 dB}
    \end{overpic}
\caption{Retaining 2,600 largest transform coefficients}
  \end{subfigure}
\par\bigskip %
 \begin{subfigure}[b]{\linewidth}
     \begin{overpic}[width=.25\linewidth]{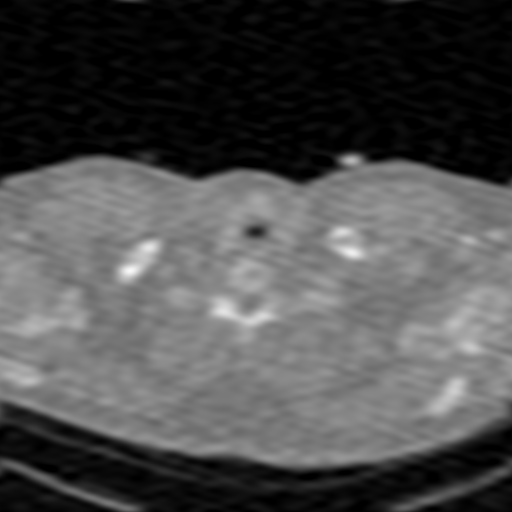}
      \put (2, 88) {\color{white} \footnotesize DFT, 18.8 dB}
    \end{overpic}%
   \begin{overpic}[width=.25\linewidth]{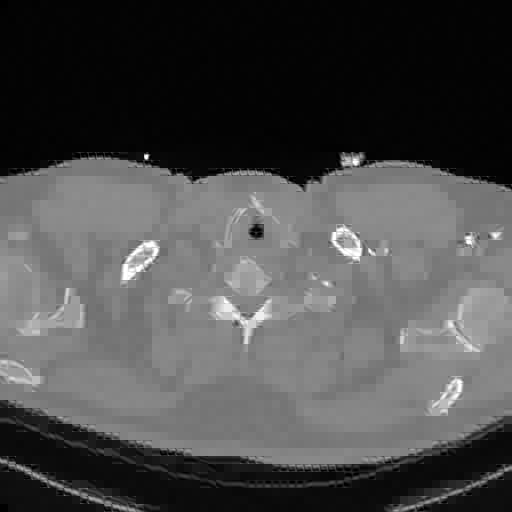}
      \put (2, 88) {\color{white} \footnotesize DCT, 22.3 dB}
    \end{overpic}%
   \begin{overpic}[width=.25\linewidth]{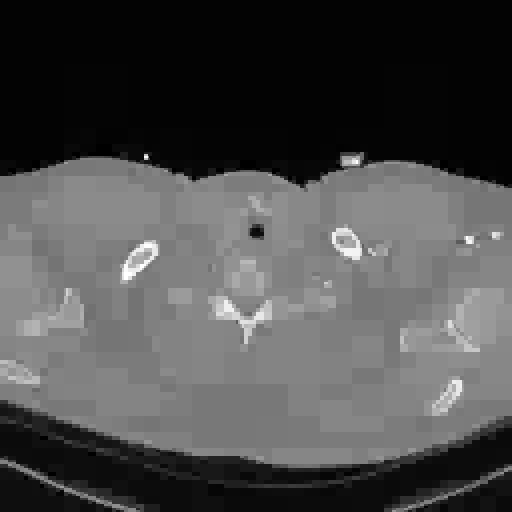}
      \put (2, 88) {\color{white} \footnotesize Haar, 20.9 dB}
    \end{overpic}%
   \begin{overpic}[width=.25\linewidth]{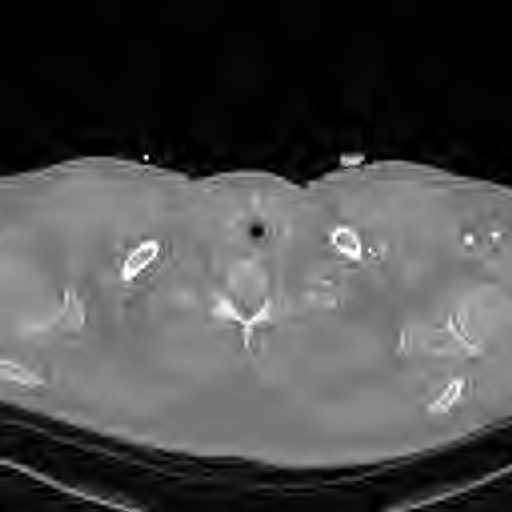}
      \put (2, 88) {\color{white} \footnotesize db4, 21.2 dB}
    \end{overpic}
\caption{Retaining 1,300 largest transform coefficients}
  \end{subfigure}
  \caption{Example of the compressibility of biomedical images
    using a \ac{CT} image from \acf{TCIA}~\cite{clark_cancer_2013}.
    We transform a 512 $\times$ 512 image using the
    discrete Fourier transform (DFT),
    the 8 $\times$ 8  block discrete cosine transform (DCT),
    the Haar wavelet transform (Haar),
    or the Daubechies 4 wavelet transform (db4).
    Then, we set the smallest 95, 99, and 99.5 percent of the coefficients to zero
    and reconstruct.
    The \ac{SNR} of the reconstruction is reported in decibels.
    At these compression ratios, much of the detail of the original is preserved,
    suggesting the image is compressible.
    On the other hand, the compressed versions are not exact copies of the original,
    and, therefore, this type of compression is usually not used when storing medical images.
  }
  \label{fig:transforms}
\end{figure}

We begin our discussion by considering $\vv{L} = \vv{I}$, 
so that $\|  \vv{f} \|_0 = K$.
Given measurements, $\vv{g} = \vv{H}\vv{f}$,
the compressive sensing problem is
\begin{equation}
  \label{eq:CS}
\argmin_\vv{f}  \|\vv{H}\vv{f} - \vv{g}\|^2_2
\quad \text{s.t.} \quad  \| \vv{f} \|_0 \le K,
\end{equation}
where the $\ell_0$ norm, $\| \vv{f} \|_0$, counts
the number of nonzero entries of $\vv{f}$.
In general, \eqref{eq:CS} is very challenging to solve,
however under certain conditions on $\vv{H}$~\cite{donoho_compressed_2006,candes_stable_2006},
the solution is unique and the problem is equivalent to 
\begin{equation}
  \label{eq:CSl1}
  \argmin_\vv{f}  \|\vv{H}\vv{f} - \vv{g}\|^2_2 + \lambda\| \vv{f} \|_1 ,
\end{equation}
where $\| \vv{f} \|_1 = \sum_n |[ \vv{f} ]_n|$,
for a particular choice of $\lambda$.
(Note the change from the $\ell_0$ to $\ell_1$ norm;
this is what makes the problem tractable.)

A more general formulation of the sparse recovery problem
includes a sparsifying transform, $\vv{L}$,
and is formulated as
\begin{equation}
  \label{eq:CS-L}
\argmin_\vv{f}  \|\vv{H}\vv{f} - \vv{g}\|^2_2
\quad \text{s.t.} \quad  \| \vv{L} \vv{f} \|_0 \le K,
\end{equation}
which, under certain conditions~\cite{candes_robust_2006},
we can equivalently reformulate into an $\ell_1$ problem,
\begin{equation}
\label{eq:l1}
\argmin_\vv{f}  \|\vv{H}\vv{f} - \vv{g}\|^2_2 + \lambda\| \vv{L} \vv{f} \|_1.
\end{equation}
While the formulation \eqref{eq:l1} is quite similar to the
Tikhonov formulation \eqref{eq:problem_tikhonov},
changing to the $\ell_1$ norm has important consequences;
we contrast these two formulations in Figures~\ref{fig:ellipse} and \ref{fig:ellipse_sparse}.

\begin{figure}[!htbp]
  \centering
\begin{tikzpicture}

  \begin{axis}[
    width = 5.5cm,
    xmin=-2, xmax=3.5,
    xlabel = $f_1$,
    ymin=-2, ymax=2,
    ylabel = $f_2$,
    ticks = none,
    axis lines = middle,
    x=1cm, y=1cm,
    axis line style={latex-latex},
     every axis y label/.style={
      at={(ticklabel* cs:1)},
      anchor=east,
    },
    ]

    \fill [name path = ell, gray, opacity=0.1] (1.75, .5) ellipse [
     x radius = 1.6180,
     y radius = 0.6180,
     rotate = 121.7175
     ];
     \fill [gray] (1.75, .5) circle (3pt);

    \draw [color=red, dashed]  (1.3385, 0) -- (0, 1.3385) -- (-1.3385, 0) -- (0, -1.3385) -- cycle;
    \fill [red] (1.3385, 0) circle (3pt);

    \coordinate (B) at (1.0363, 0.5133);
    \node [draw, circle through=(B), color=blue] at (0,0) {};
    \fill [blue] (B) circle (3pt);

  \end{axis}

\end{tikzpicture}  
\caption{Comparison of $\ell_2$ and $\ell_1$ regularization
when the number of linearly independent measurements is not less than the dimension of the reconstruction.
The gray point is the unique unregularized solution and
the gray ellipse is $\{\vv{f} :  \|\vv{H}\vv{f} - \vv{g}\|_2
^2 \le \sigma^2\}$.
The solid blue line is a level set of $\|\vv{f}\|_2$,
and the blue point marks the solution of the $\ell_2$-regularized problem.
Likewise, the dashed red line is a level set of $\|\vv{f}\|_1$
and the red point marks the $\ell_1$ solution.
As expected, the $\ell_1$ solution is sparse (because it has only one non-zero element, $f_1$).
}
\label{fig:ellipse}
\end{figure}
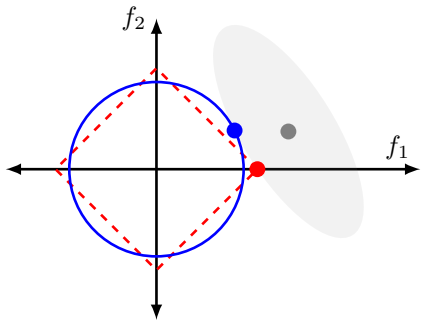

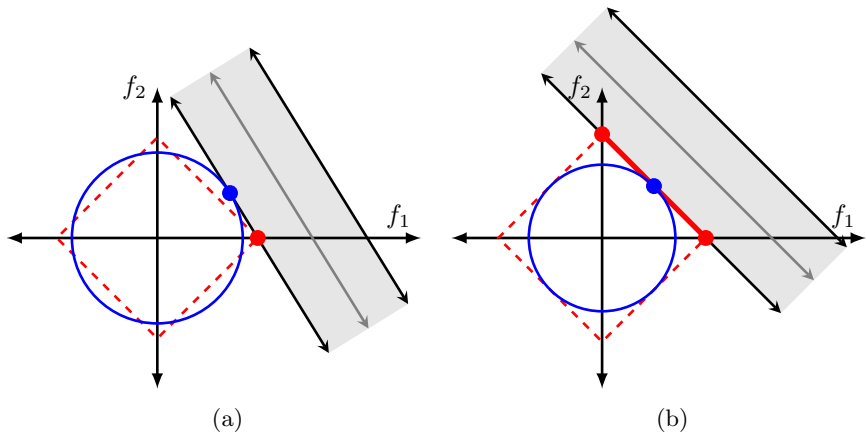
\begin{figure}[htbp]
  \centering
  \begin{subfigure}[b]{.5\linewidth}
    \begin{tikzpicture}
      \begin{axis}[
        clip=false,
        width = 5.5cm,
        xmin=-2, xmax=3.5,
        xlabel = $f_1$,
        ymin=-2, ymax=2,
        ylabel = $f_2$,
        ticks = none,
        axis lines = middle,
        every axis y label/.style={
          at={(ticklabel* cs:1)},
          anchor=east,
        },
                x=1cm, y=1cm,
        axis line style={latex-latex}
        ]

\def\w{.618cm}

\coordinate (center) at (1.75, .5);
\def\rot{121.7175}

\path [rotate=\rot, fill, opacity=0.1]
let \p1 = (center) in
(\x1 - 2cm, \y1 -.618cm) rectangle (\x1 + 2cm, \y1 + .618cm);

\draw [rotate=\rot, <->, name path=ell]
let \p1 = (center) in
(\x1-2cm , \y1 + .618cm) -- (\x1 + 2cm, \y1 + .618cm); 

\draw [rotate=\rot, <-> ]
let \p1 = (center) in
(\x1-2cm , \y1 - .618cm) -- (\x1 + 2cm, \y1 - .618cm); 

\draw [rotate=\rot, <->, gray]
let \p1 = (center) in
(\x1-2cm , \y1) -- (\x1 + 2cm, \y1);

\draw [color=red, dashed]  (1.3325, 0) -- (0, 1.3325) -- (-1.3325, 0) -- (0, -1.3325) -- cycle;
\fill [red] (1.3325, 0) circle (3pt);

\coordinate (B) at (.9642, .5959);
\node [draw, circle through=(B), color=blue] at (0,0) {};
\fill [blue] (B) circle (3pt);
\end{axis}
\end{tikzpicture} 
\caption{}
\end{subfigure}%
\begin{subfigure}[b]{.5\linewidth}
  \begin{tikzpicture}
     \begin{axis}[
        clip=false,
        width = 5.5cm,
        xmin=-2, xmax=3.5,
        xlabel = $f_1$,
        ymin=-2, ymax=2,
        ylabel = $f_2$,
        ticks = none,
        axis lines = middle,
        every axis y label/.style={
          at={(ticklabel* cs:1)},
          anchor=east,
        },
        x=1cm, y=1cm,
        axis line style={latex-latex}
        ]

\def\w{.618cm}

\coordinate (center) at (1.75, .5);
\def\rot{135}

\path [rotate=\rot, fill, opacity=0.1]
let \p1 = (center) in
(\x1 - 1.5cm, \y1 -.618cm) rectangle (\x1 + 3cm, \y1 + .618cm);

\draw [rotate=\rot, <->, name path=ell]
let \p1 = (center) in
(\x1-1.5cm , \y1 + .618cm) -- (\x1 + 3cm, \y1 + .618cm); 

\draw [rotate=\rot, <-> ]
let \p1 = (center) in
(\x1-1.5cm , \y1 - .618cm) -- (\x1 + 3cm, \y1 - .618cm); 

\draw [rotate=\rot, <->, gray]
let \p1 = (center) in
(\x1-1.5cm , \y1) -- (\x1 + 3cm, \y1);

\coordinate (A) at (0, 1.3760);
\coordinate (B) at (1.3760, 0);
\draw [color=red, dashed]  (1.3760, 0) -- (0, 1.3760) -- (-1.3760, 0) -- (0, -1.3760) -- cycle;
\fill [red] (A) circle (3pt);
\fill [red] (B) circle (3pt);
\draw [color=red, line width=2pt] (A) -- (B);

\coordinate (B) at (1.125-.437, 1.125-.437); %
\node [draw, circle through=(B), color=blue] at (0,0) {};
\fill [blue] (B) circle (3pt);
\end{axis}
\end{tikzpicture} 
\caption{}\label{fig:ellipse_sparse:nonunique}
\end{subfigure}  
\caption{(a) Effect of regularization
  when the number of linearly independent measurements is less than the size of the reconstruction.
  In contrast to Figure~\ref{fig:ellipse}, 
  the unregularized solution is nonunique and, specifically,
  is an affine subspace (gray line).
  Otherwise, the situation is the same as in Figure~\ref{fig:ellipse},
  with a nonsparse $\ell_2$ solution and sparse $\ell_1$ solution.
  (b) For certain problems, the $\ell_1$ solution is nonunique
  (red line segment).
  Even in these cases, Theorem~\ref{thm:l1-representer} states that
  the extreme points of the solution set (red points)
  are sparse.
}
\label{fig:ellipse_sparse}
\end{figure}

The formulation \eqref{eq:l1} is called the analysis form of the regularization,
because the matrix $\vv{L}$ retrieves the sparse coefficients from the original signal.
This is in contrast to the synthesis form,
\begin{equation}
  \argmin_{\vv{\alpha} \in \mathbb{R}^N} \| \vv{H} \tilde{\vv{L}} \vv{\alpha} - \vv{g} \|^2_2 + \lambda \| \vv{\alpha} \|_1,
  \label{eq:sparse-synth}
\end{equation}
where the matrix $\tilde{\vv{L}}$ now acts to construct the signal from its sparse coefficients.
If the synthesis transform has a left inverse, this is merely a change in notation;
otherwise, the two forms are meaningfully different.
For more discussion of these two forms, see \cite{elad_analysis_2007}.

\section{Representer Theorems for \texorpdfstring{$\ell_2$}{l2} and \texorpdfstring{$\ell_1$}{l1} Problems}
Another perspective on sparsity-promoting regularization is given by representer theorems
that specify the form of solutions to certain minimization problems.
For example, for Tikhonov regularization with $\vv{L}$ being the identity matrix,
we can state the following representer theorem,
which is a simplified special case of a result from \cite{unser_representer_2016}.
\begin{theorem}[Convex Problem with $\ell_2$ Minimization]
The problem
\begin{equation*}
  \argmin_{\vv{f}} \| \vv{f} \|^2_2 
\quad \text{s.t.} \quad \| \vv{H}\vv{f} - \vv{g} \|^2_2 \le \sigma^2,
\end{equation*}
has a unique solution of the form
\begin{equation*}
  \tilde{\vv{f}} = \vv{H}^* \vv{a},
\end{equation*}
for a suitable set of coefficients, $\vv{a} \in \mathbb{R}^M$.
\label{thrm:l2-representer}
\end{theorem}
The useful insight here is that the solution to an $\ell_2$-regularized problem
always has the form of a weighted sum of the original measurement vectors,
i.e., the columns of $\vv{H}$.
Moreover, the number of elements in the sum is equal to the number of measurements, $M$,
so there is no reason to expect $\vv{f}$ to be sparse,
unless these measurement vectors are themselves sparse in some transform domain.
If we choose instead to minimize a function
including an invertible regularization operator, $\| \vv{L} \vv{f} \|^2_2$,
we can state a similar theorem with $\vv{H}^*$ replaced
by $(\vv{L}^* \vv{L})^{-1} \vv{H}^*$.

Contrast Theorem~\ref{thrm:l2-representer} with the $\ell_1$ representer theorem~\cite{unser_representer_2016},
\begin{theorem}[Convex Problem with $\ell_1$ Minimization]
The set
\begin{equation*}
 \mathcal{V} =  \argmin_{\vv{f}} \| \vv{f} \|_1 
\quad \text{s.t.} \quad \| \vv{H}\vv{f} - \vv{g} \|^2_2 \le \sigma^2,
\end{equation*}
is convex, compact, and has extreme points of the form
\begin{equation*}
  \tilde{\vv{f}} = \sum_{k=1}^K [\vv{a}]_k \vv{e}_{[\vv{n}]_k},
\end{equation*}
where $\{\vv{e}_{n}\}_{n=1}^N$ is the standard basis for $\mathbb{R}^N$
(i.e., unit vectors pointing along each of the axes)
and for a suitable set of coefficients $\vv{a} \in \mathbb{R}^K$
and locations $\vv{n} \in \{1, 2, \dots, N \}^K$
with $K \le M$.
\label{thm:l1-representer}
\end{theorem}
Looking at the form of $\tilde{\vv{f}}$,
we see that it is sparse:
it has fewer nonzero terms
than the number of measurements.
The amplitudes of its nonzero terms are 
given by $\vv{a}$
and their locations are given by $\vv{n}$.
One complication is that the solution to the $\ell_1$
problem is not, in general, unique
(though with additional conditions it is).
Thus, the theorem is stated in terms of the extreme points
of the solution set, i.e., those solutions that cannot be expressed
as linear combinations of other solutions.
When $\vv{L}$ is not the identity,
a similar theorem can be stated where
the extreme points of the solution set are
built out of a sparse linear combination of dictionary vectors that depend on $\vv{L}$~\cite{unser_representer_2016}.

\section{Bayesian View}
We can turn to a statistical perspective on sparsity-promoting regularization,
basing the regularization term on a statistical model of the signal.
When we expect the signal model to be accurate,
this approach allows us to design reconstruction algorithms that are optimal
with regard to chosen statistical criteria.
Even when the signal model is only an educated guess,
the statistical formulation can be helpful in designing new regularization terms with good empirical performance.
Here, we follow the approach of~\cite{unser_introduction_2014},
which places both $\ell_1$ and $\ell_2$ regularization along a spectrum of sparse regularizers.

Working for the moment in 1D,
we model the signal as a stochastic process specified by the innovation model
\begin{equation}
  \label{eq:u-equals}
  \vv{u} = \vv{L} \vv{f},
\end{equation}
where the discrete innovation, $\vv{u}$, has \ac{iid} elements
with an infinitely divisible \ac{PDF}, $p_U$,
and $\vv{L}$ is a whitening operator.
The restriction to this family is motivated by the property that
this is the only configuration where \eqref{eq:u-equals} has a
continuous-domain counterpart as a stochastic differential equation,
$w = L \{f\}$, that specifies the family of sparse stochastic processes~\cite{unser_introduction_2014}.
Yet,
the family of infinitely divisible \ac{PDF}s is still rather larger:
it includes the Gaussian distribution,
as well as several sparser (heavy-tailed or with more mass at the origin)
distributions, e.g., the Laplace, compound-Poisson, and Cauchy.

Combining this signal model with a forward operator $\vv{H}$ and
a Gaussian noise model, we can express the posterior
\begin{align}
  p(\vv{f} \mid \vv{g}) 
  &= \frac{p(\vv{g} \mid \vv{f})  p( \vv{f} ) }{ p(\vv{g}) }\\
  &= \frac{p(\vv{g} - \vv{H} \vv{f})  p( \vv{f} ) }{ p(\vv{g}) } \\
  &\propto \exp \left(
    -\frac{\| \vv{g} - \vv{H} \vv{f} \| }{2 \sigma^2}
 \right)
    \prod_{n = 1}^N p_U ( [\vv{L} \vv{f}]_n ).
\end{align}
Taking the negative log and minimizing (i.e., performing \ac{MAP} estimation), we have
\begin{equation}
  \tilde{\vv{f}} = \argmin_\vv{f} 
  \frac{1}{2}  \| \vv{g} - \vv{H} \vv{f} \|^2_2 
  +  \sigma^2 \sum_n \Phi_U ( [\vv{L} \vv{f}]_n ),
\label{eq:MAP-sparse}
\end{equation}
where $\Phi_U(u) = - \log p_U (u) $ is called the potential function.

\begin{table}[!htbp]
  \centering
  \begin{tabular}{@{}ccc@{}}\toprule
    & $p_U(u)$ & $\Phi_U(u)$ \\ \midrule
    Gaussian & $\frac{1}{\sqrt{2 \pi \sigma^2}} e^{\frac{-u^2}{2\sigma^2}}$ &
             $\frac{1}{2\sigma^2}u^2 + C$   \\
Laplace & $\frac{\lambda}{2}e^{-\lambda |u|}$ &
$\lambda | u | + C$ \\
Student & $\frac{1}{B(r, \frac{1}{2})} \left( \frac{1}{u^2 + 1} \right)^{r+\frac{1}{2}}$ & 
$(r + \frac{1}{2}) \log (1+ u^2) + C$\\
\bottomrule
  \end{tabular}
  \caption{Infinitely divisible distributions and their corresponding potential functions.
  $B$ is the beta function.}
\label{tab:infinite-divisible}
\end{table}

We give a few examples of infinite divisible distributions 
and the corresponding potential functions in Table~\ref{tab:infinite-divisible}.
In the first line, note that the Gaussian potential
corresponds to the squared Euclidean norm (with scaling and plus a constant);
thus, $\ell_2$ regularization corresponds to \ac{MAP} estimation
of a stochastic process with a Gaussian innovation
($\vv{u}$ in \eqref{eq:u-equals}).
This is a generalization of the example from Section~\ref{sec:classical-stat}
because, here, our model if $\vv{f}$ involves the operator $\vv{L}$.
In the second line, we see that the Laplace potential corresponds
to the $\ell_1$ norm,
which means that $\ell_1$ regularization corresponds to $\ac{MAP}$
estimation of a stochastic process with a Laplace-distributed innovation.
This links with the concept of sparsity because the Laplace distribution
is sparser than the Gaussian.
The Student distribution (third line of Table~\ref{tab:infinite-divisible})
is the sparsest of the three,
and, again, its potential function is closely related to sparsity-promoting
regularization.
Specifically, it is an upper bound on the $\log$ prior,
\begin{equation}
  \sum_{n=1}^N \log | [\vv{f}]_n |,
\end{equation}
which can be used as a relaxation of the $\ell_0$ norm
in certain compressive sensing problems
without changing the minima~\cite{wipf_iterative_2010}.

\section{Algorithms}
\label{sec:sparse-algo}
A wide range of algorithms have been developed
for solving problem like \eqref{eq:l1}.
One key concept for understanding them is that of the proximal operator.
\begin{definition}
The \emph{proximal operator}, $\prox_\Phi: \mathbb{R}^N \to \mathbb{R}^N$,  of a convex function, $\Phi:  \mathbb{R}^N \to \mathbb{R}$, is defined by
\begin{equation}
  \prox_\Phi (\vv{u}) = \argmin_{\vv{f}} \frac{1}{2} \| \vv{u} - \vv{f} \|^2_2 + \Phi(\vv{f}).
\end{equation}
  Intuitively, it finds an $\vv{f}$ with a low $\Phi(\vv{f})$,
while being close (proximal) to $\vv{u}$.
\end{definition}
If $\Phi$ acts pixel by pixel,
then the optimization problem that the $\prox$ solves
can be decomposed into $N$ scalar problems, one for each pixel.
And, in many cases, these scalar problems can be solved efficiently
and in parallel.

The $\prox$ is used in a class of algorithms called
forward-backward splitting~\cite{combettes_signal_2005}
or \ac{ISTA}~\cite{beck_fast_2009}
that solve problems of the form
\begin{equation}
  \argmin_{\vv{f}} J_1(\vv{f}) + J_2(\vv{f}),
\end{equation}
with $J_1$ differentiable and $J_2$ having a fast $\prox$.
For example, these approaches can solve
the sparse reconstruction problem \eqref{eq:l1} when $\vv{L}$ is the identity.
They achieve this by alternating a step down the gradient of $J_1$
with an application of the $\prox$ of $J_2$,
which can be shown to converge.
In fact, this method is closely related to projected gradient descent
(discussed in Section~\ref{sec:iterative-reconstruction})
because when a constraint is expressed as a regularization functional
(which takes the value infinity whenever the constraint is violated)
the associated $\prox$ is just projection onto the constraint region.

When the operator $\vv{L}$ is involved,
the sparse reconstruction problem is typically much more
difficult to solve.
A useful paradigm is to split the data term and regularization terms
by introducing an auxiliary variable.
Using the sparse \ac{MAP} formulation \eqref{eq:MAP-sparse} as an example,
splitting results in 
\begin{equation}
  \tilde{\vv{f}} = \argmin_{\vv{f}, \vv{u}}
  \frac{1}{2}  \| \vv{g} - \vv{H} \vv{f} \|^2_2 
  +  \sigma^2 \sum_n \Phi_U ( [\vv{u}]_n )
\quad \text{s.t.} \quad \vv{u} = \vv{L} \vv{f}.
\end{equation}
This problem is equivalent to the original one,
but written in the form required by the \ac{ADMM}~\cite{boyd_distributed_2011}.
The \ac{ADMM} then specifies an iterative algorithm,
\begin{align}
 \label{eq:ADMM-1}
  \vv{f}^{(k+1)} &= \argmin_{\vv{f}}
  \frac{1}{2}  \| \vv{g} - \vv{H} \vv{f}^{(k)} \|^2_2  +
  (\rho / 2) \|\vv{L} \vv{f} - \vv{u}^{(k)} \|^2_2 \\
\label{eq:ADMM-2}
  \vv{u}^{(k+1)} &= \argmin_{\vv{u}}  \sigma^2 \sum_n \Phi_U ( [\vv{u}]_n ) +
(\rho / 2) \|\vv{L} \vv{f}^{(k+1)} - \vv{u} \|^2_2 \\
  \label{eq:ADMM-3}
  \vv{\alpha}^{(k+1)} &= \vv{\alpha}^{(k)} + \vv{L} \vv{f}^{(k+1)} - \vv{u}^{(k+1)}.
\end{align}

It turns out that these three problems are significantly simpler to solve
than the original.
The update \eqref{eq:ADMM-1} is a quadratic problem
that can be solved using the techniques from classical image reconstruction
(Section~\ref{sec:iterative-reconstruction}), 
e.g., the \ac{CG} algorithm.
For the update \eqref{eq:ADMM-2} we can use the $\prox$
of $\Phi_U$, which is known in closed-form for many potential functions.
Finally, the update \eqref{eq:ADMM-3} only requires application of $\vv{L}$.

\section{Summary}
Modern reconstruction techniques focus on sparsity-promoting regularization,
with the goal of reconstructing images from very few measurements.
This is possible only when the image to be recovered can be sparsely represented
in some transform domain.
See Figure~\ref{fig:comparison_l2_l1} for a comparison between
classical variational reconstructions and modern sparsity-based ones.
The algorithms to solve sparse reconstruction problems often involve
splitting the objective function;
each iteration of these algorithms requires solving a classical reconstruction problem,
which makes them computationally heavy.

\begin{figure}[!htbp]
  \centering
   \begin{subfigure}{1.0\linewidth}
  \begin{subfigure}{.328\linewidth}
    \includegraphics[width=\linewidth]{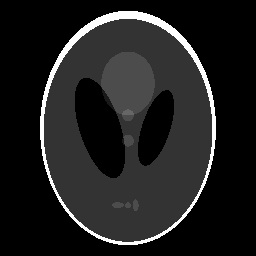}
    \caption{$\vv{f}$}
   \end{subfigure}%
   \begin{subfigure}{.328\linewidth}
     \includegraphics[width=\linewidth]{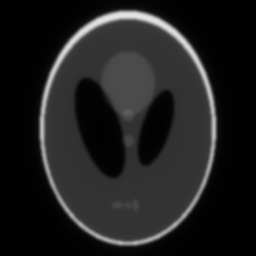}
     \caption{$\vv{C}\vv{f}$}
   \end{subfigure}%
     \begin{subfigure}{.328\linewidth}
       \includegraphics[width=\linewidth]{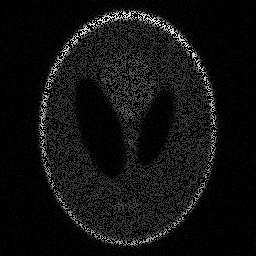}
       \caption{$\vv{M} \vv{C} \vv{f} + \vv{n}$}
     \end{subfigure}
     \end{subfigure}\\
  \begin{subfigure}{1.0\linewidth}
   \begin{overpic}[width=.1666\linewidth, trim={0 0 128px 0}, clip]{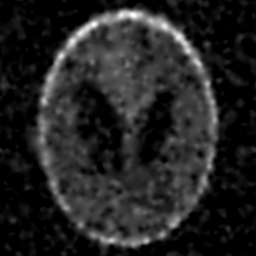}
      \put (2, 90) {\color{white} \footnotesize 5.1 dB}
    \end{overpic}%
     \begin{overpic}[width=.1666\linewidth, trim={128px 0 0 0}, clip]{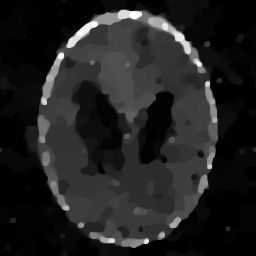}
      \put (20, 2) {\color{white} \footnotesize 5.9 dB}
    \end{overpic}%
      \begin{overpic}[width=.1666\linewidth, trim={0 0 128px 0}, clip]{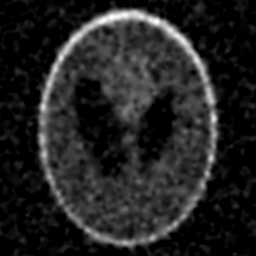}
      \put (2, 90) {\color{white} \footnotesize 5.9 dB}
    \end{overpic}%
     \begin{overpic}[width=.1666\linewidth, trim={128px 0 0 0}, clip]{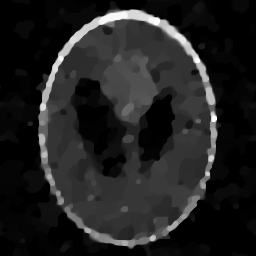}
      \put (20, 2) {\color{white} \footnotesize 7.2 dB}
    \end{overpic}%
      \begin{overpic}[width=.1666\linewidth, trim={0 0 128px 0}, clip]{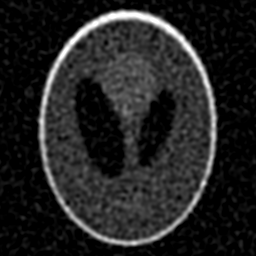}
      \put (2, 90) {\color{white} \footnotesize 7.9 dB}
    \end{overpic}%
     \begin{overpic}[width=.1666\linewidth, trim={128px 0 0 0}, clip]{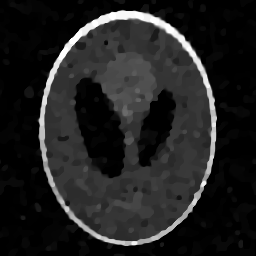}
      \put (20, 2) {\color{white} \footnotesize 10.0 dB}
    \end{overpic}%
    \caption{5 dB SNR measurements, retaining 5, 10, and 50\% of pixels.
      Left side of each image: $\ell_2$ regularization.
    Right side of each image: $\ell_1$ regularization.}
  \end{subfigure}\\
  \begin{subfigure}{1.0\linewidth}
   \begin{overpic}[width=.1666\linewidth, trim={0 0 128px 0}, clip]{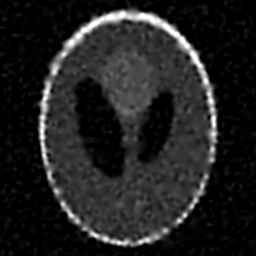}
      \put (2, 90) {\color{white} \footnotesize 7.5 dB}
    \end{overpic}%
     \begin{overpic}[width=.1666\linewidth, trim={128px 0 0 0}, clip]{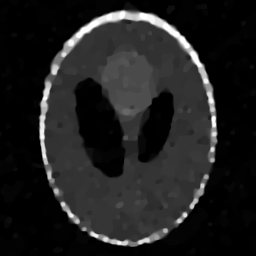}
      \put (20, 2) {\color{white} \footnotesize 8.5 dB}
    \end{overpic}%
      \begin{overpic}[width=.1666\linewidth, trim={0 0 128px 0}, clip]{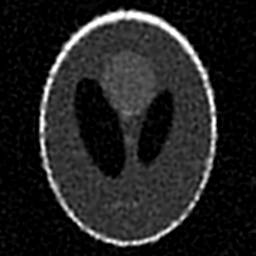}
      \put (2, 90) {\color{white} \footnotesize 8.8 dB}
    \end{overpic}%
     \begin{overpic}[width=.1666\linewidth, trim={128px 0 0 0}, clip]{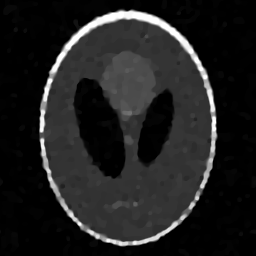}
      \put (20, 2) {\color{white} \footnotesize 10.5 dB}
    \end{overpic}%
      \begin{overpic}[width=.1666\linewidth, trim={0 0 128px 0}, clip]{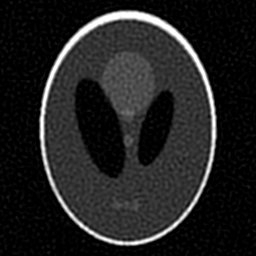}
      \put (2, 90) {\color{white} \footnotesize 11.0 dB}
    \end{overpic}%
     \begin{overpic}[width=.1666\linewidth, trim={128px 0 0 0}, clip]{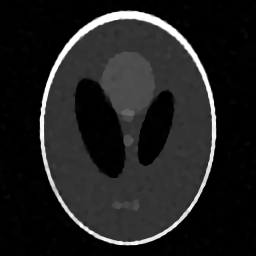}
      \put (20, 2) {\color{white} \footnotesize 14.5 dB}
    \end{overpic}%
    \caption{20 dB SNR measurements, retaining 5, 10, and 50\% of pixels.
     Left side of each image: $\ell_2$ regularization.
    Right side of each image: $\ell_1$ regularization.}
  \end{subfigure}
  \caption{Comparison of classical regularization
    and sparsity-promoting regularization.
    A discrete test image~\cite{shepp_maximum_1982} (a) is degraded with blur (b),
    subsampling, and noise (c).
    For two different levels of measurement noise (rows (d) and (e)),
    we reconstruct with $\ell_2$ and $\ell_1$ regularization (
    left and right sides of each panel of (d) and (e), each respectively),
    with regularization parameter $\lambda$ chosen to maximize \ac{SNR}.
    The \ac{SNR} of the reconstruction is reported in decibels for reconstruction.
    Qualitatively, classical regularization fills gaps by smoothing,
    while sparsity-promoting regularization gives a piecewise constant solution.
    For this image, the latter is superior.
  }
  \label{fig:comparison_l2_l1}
\end{figure}

\subsection{Further Reading}
Our discussion of regularization is a simplified version of a much more complex story;
see \cite{engl_regularization_1996} for a textbook treatment and
\cite{benning_modern_2018} for more historical details and mathematical depth.
See \cite{candes_sparse_2007} for an introduction to compressive sensing,
and \cite{rani_systematic_2018} for an overview of the fields where it is applied.
For examples of the statistical view presented here applied to real problems,
see~\cite{unser_introduction_2014}.

\chapter{The Learning (R)Evolution}
\label{chap:learning}
The field of machine learning,
which deals with creating computer programs that can improve with training~\cite{mitchell_machine_1997},
has existed for decades
and has even been used for solving image reconstruction problems
as early as in the 1980s~\cite{zhou_image_1988}.
Beginning with the emergence of one particular machine learning algorithm,
the \ac{CNN},
as a powerful and practical tool for object recognition in 2012~\cite{russakovsky_imagenet_2015},
there has been a surge in interest in applying learning-based methods
to a wide variety of problems in image processing and computer vision,
including image reconstruction.
Our goal in this section is to give a broad overview of these approaches.

So far, we have considered a setting where we have measurements, $\vv{g}$,
and an understanding of the measurement device
from which we build a forward model, $\vv{H}$.
Additionally, we may have some knowledge of the image we want to reconstruct,
allowing us to design a variational regularization term.
In the learning formulation, we assume that we also have access to
training data:
$T$ pairs of measurements and
their corresponding ground truth reconstructions, $\{\vv{g}_t, \vv{f}_t \}_{t=1}^T$.
The inclusion of training data changes the task of the engineer from
designing a reconstruction procedure that maps measurements to images,
$\mathcal{R} : \vv{g} \mapsto \vv{f}$,
to designing a learning procedure
that maps training sets to reconstruction procedures,
$\mathcal{L} : \{\vv{g}_t, \vv{f}_t \}_{t=1}^T \mapsto \mathcal{R}$.

We have organized this section along a spectrum of how much is learned.
On one end of the spectrum are methods where the training data
is used to improve a part of a variational method,
usually the forward model or the regularization term.
On the other end are  \emph{pipeline} or \emph{end-to-end} methods,
where most or all of the reconstruction procedure is a generic regression function
with its parameters learned during training.
We also discuss a few methods that do not fit into our learning formulation
as well as the problem of how to acquire training data.

\section{Learning the Forward Model}
Instead of relying on a physical model of the imaging setup alone
to design the forward model, $\vv{H}$,
we can use training data to estimate the forward model.
This approach is advantageous because the accuracy of the forward model
has a large impact on the quality of the reconstruction.
Learning the forward model is a daunting task for large problems
because the number of elements of the matrix is
the number of image pixels times the number of measurements%
---roughly the number of image pixels squared,
easily terabytes of data.
Often, knowledge of the structure of the forward model
can greatly reduce the number of measurements needed to estimate it.

One example~\cite{shaw_point_1991} comes from deconvolution of microscopy images.
Using our knowledge of the underlying physics,
we can assume that the system matrix is a convolution---%
a Toeplitz (diagonal-constant) matrix.
If we further assume that the impulse response
(known in microscopy as the \ac{PSF})
is sufficiently limited in space,
then the system matrix is completely specified by its middle rows.
By imaging small beads, we can create
a training set with images of the form
$\vv{f}_t =
\begin{bmatrix}
  0 & \dots & 0 & 1 & 0 & \dots & 0 
\end{bmatrix}^T$.
The resulting measurements, $\{\vv{g}^t\}$, each
provide an independent estimate of the \ac{PSF}
that can be averaged to reduce the effect of noise.

Another example~\cite{panin_fully_2006}, comes from \ac{PET}
where measurements of a carefully positioned point source
were used to estimate the rows of the forward model.
While estimating the entire forward model would require
approximately four million measurements,
the authors were able to use the geometrical symmetries of the system
plus a physical model to interpolate the forward model from
just 1,599 measurements.

A final example comes from \ac{MRI},
where coil sensitivity maps are a key piece of the forward model
that must be estimated from data.
The standard approach is to collect data from a body coil,
which has uniform sensitivity,
and use it to estimate the sensitivity maps.
The downside of this approach is that the \ac{SNR} of the body coil data is low,
making the sensitivity estimation noisy.
The problem remains an area of active research~\cite{fessler_model_2010}.

\section{Learning the Regularization Term}
There are a variety of ways to use training data to improve
the regularization term.
One straightforward approach is to use training to
optimally adjust the regularization strength,
i.e., tune the parameter $\lambda$.
In the approach of \cite{kunisch_bilevel_2013},
the training takes the form of a bilevel optimization,
\begin{equation}
  \argmin_\lambda \sum_{t=1}^T \| \mathcal{R}_\lambda(\vv{g}_t) - \vv{f}_t\|^2_2,
\end{equation}
with
\begin{equation}
  \mathcal{R}_\lambda(\vv{g}) = \argmin_f \| \vv{g} - \vv{H} \vv{f} \|^2_2
  + \lambda \| \vv{L} \vv{f} \|_p^p,
\end{equation}
where $p$ is one or two.
The solution of this problem is the value of $\lambda$ 
that gives optimal performance on the training set.
Reconstructions of new images should then use this optimal $\lambda$.
The same method can be extended multiple $\vv{L}$ operators,
each with their own weight;
e.g., \cite{kunisch_bilevel_2013} uses the set of $5 \times 5$ \ac{DCT} filters.

Going a step further,
we can try to learn the form of the potential function,
i.e., the  $\Phi$ in a regularization term $\Phi( \vv{L} \vv{f} )$.
Recalling the \ac{ADMM} algorithm from Section~\ref{sec:sparse-algo},
we know that solving a regularized inverse problem
involves applying the proximal operator associated with the potential function.
We can implicitly specify the potential function by learning the proximal operator,
e.g., by parameterizing it using 1D B-splines~\cite{kamilov_learning_2016,nguyen_learning_2018}.
We can see $\prox$ learning as a generalization of tuning the regularization weight:
instead of scaling a known function,
we deform a parametric function.
If the learned functions satisfy mild constraints,
the resulting algorithm is guaranteed to converge.

Or, as a different approach,
we can learn the filters in the regularization term
leading to a learning procedure~\cite{yaghoobi_constrained_2013},
\begin{equation}
  \argmin_{\vv{L} \in \mathcal{C}} \sum_{t=1}^T \| \vv{L} \vv{f}_t\|_1,
\end{equation}
where the matrix $\vv{L}$ must be constrained to lie in a set $\mathcal{C}$
to avoid the trivial solution $\vv{L} = \vv{0}$.
A variation of this approach is to allow for noisy training data,
leading to
\begin{equation}
  \argmin_{\vv{L} \in \mathcal{C}, \{\vv{h}\}_t} \sum_{t=1}^T \| \vv{L} \vv{h}_t\|_1
  \quad \text{s.t.} \quad \| \vv{h}_t - \vv{f}_t \| \le \sigma.
\end{equation}
Recalling the connection between the $\ell_1$ norm and sparsity,
we can view this learning procedure as seeking a sparsifying transform
for the images in the training set.
Intuitively, this learned transform should allow for sparser representations
(and therefore better regularization)
than a fixed one.

We can also learn a sparsifying transform in the synthesis formation.
Doing this is known as sparse dictionary learning (or sparse coding)~\cite{elad_image_2006}.
To learn the $\vv{L}$ in a synthesis formulation \eqref{eq:sparse-synth}, we solve
\begin{equation}
  \argmin_{\vv{L}, \{\vv{\alpha}\}_t} \sum_{t=1}^T \| \vv{L} \vv{\alpha}_t - \vv{f}_t\|^2_2 + \| \vv{\alpha} \|_p^p,
  \label{eq:DL}
\end{equation}
where $p$ is zero or one.
Although for simplicity we have written the data fidelity term at the whole-image level,
it is usually evaluated at the patch level,
so that $\vv{L}$ is a dictionary that can sparsely represent, e.g., $15 \times 15$ images patches.
These patches must then be combined, usually by averaging, to create an image.
Convolutional sparse coding, which more-closely matches the formulation \eqref{eq:DL}
has also been explored~\cite{bristow_fast_2013},
as has a patch-based method with additional image-level regularization
to smooth patch boundaries~\cite{soltani_tomographic_2017}.
See \cite{xu_low_2012} for an example of this approach used in X-ray \ac{CT}.

Finally,
we can learn the entire regularization term.
Again, this can be done using a bilevel optimization~\cite{chen_insights_2014} %
\begin{equation}
  \argmin_{\vv{\theta}} \sum_{t=1}^T \| \mathcal{R}_{\vv{\theta}}(\vv{g}_t) - \vv{f}_t\|^2_2,
\end{equation}
with
\begin{equation}
  \label{eq:learned-reg}
  \mathcal{R}_\theta(\vv{g}) = \argmin_f \| \vv{g} - \vv{H} \vv{f} \|^2_2
  + \phi_{\vv{\theta}}( \vv{L}_{\vv{\theta}} \vv{f}),
\end{equation}
where $\vv{\theta}$ is a vector of parameters that controls both
the penalty function $\phi$
and the analysis filters, $\vv{L}$.
Regularizers of other forms can also be learned,
e.g., \cite{dave_solving_2019} use a pixel-wise autoregressive model
as a regularizer.

As we saw in Section~\ref{sec:sparse-algo},
certain algorithms for solving regularized reconstruction problems
like \eqref{eq:learned-reg} make use of the proximal operator of the regularization term.
One very flexible method for creating new regularizers is to replace
the proximal operator in these algorithms with, e.g., a \ac{CNN};
this implicitly defines a new regularization term.
The approach is sometimes called plug-and-play regularization.
For examples of work along these lines, see \cite{romano_little_2017,aggarwal_modl_2019,sun_online_2019,dong_denoising_2019}.
A further generalization of this idea is consensus equilibrium~\cite{buzzard_plug_2018},
which seeks reconstructions that balance data fidelity and regularization
without necessarily minimizing an optimization functional.

Another approach~\cite{chang_one_2017,gupta_cnn_2018,raj_gan_2019},
which is closer in spirit to a constrained minimization problem,
is to learn a image-to-image regression function, $\mathcal{R}_{\vv{\theta}}$,
typically a \ac{CNN},
that maps each $\vv{f}$ to the nearest member
of the set of plausible images, $\mathcal{S}$.
This function allows us to solve
\begin{equation}
  \argmin_{\vv{f} \in \mathcal{S}} \| \vv{H} \vv{f} - \vv{g} \|,
\end{equation}
which is conceptually attractive because it
returns the image that best fits the measurements among all plausible images.
The challenge here is that \emph{plausible} is defined by the training data,
which may not be sufficiently representative to capture the notion.
Yet another approach~\cite{wu_iterative_2017} is to combine the ideas of sparsity and projection
by a learning parametric image (or patch) encoder and decoder pair
according to
\begin{equation}
  \argmin_{\vv{\theta}} \sum_{t=1}^T \| \vv{f}_t - \mathcal{D}_{\vv{\theta}}( \mathcal{E}_{\vv{\theta}}( \vv{f}_t) ) \|^2_2
  \quad \text{s.t.} \quad
  \| \mathcal{E}_{\vv{\theta}}( \vv{f}_t) \|_0 \le K.
\end{equation}
Reconstruction can then be achieved via
\begin{equation}
  \argmin_{\vv{f}}  \| \vv{H}\vv{f} - \vv{g} \|^2_2 + \lambda \| \vv{f} -  \mathcal{D}_{\vv{\theta}}( \mathcal{E}_{\vv{\theta}}( \vv{f}) )  \|^2_2
  \quad \text{s.t.} \quad
  \| \mathcal{E}_{\vv{\theta}}( \vv{f}) \|_0 \le K.
\end{equation}

\section{Going Outside the Variational Framework}
One alternative to the variational formulation
that has been intensely explored in the last few years
is to perform a linear reconstruction
followed by a learned image-to-image regressor,
typically a \ac{CNN}.
The linear reconstructions often result in heavy artifacts
when the number of measurements is low
(streaks in \ac{CT} or aliasing in \ac{MRI});
the job of the regressor is to remove these artifacts;
therefore the method is sometimes referred to as
learned denoising or artifact removal.
Learning takes the form
\begin{equation}
  \argmin_{\vv{\theta}} \sum_{t = 1}^{T}
  \| \mathcal{R}_{\vv{\theta}} ( \vv{H}^\dagger \vv{g}_t ) - \vv{f}_t \|^2_2,
\end{equation}
where $\vv{H}^\dagger$ is a matrix that is an approximate inverse
of the forward model, e.g., back projection, $\vv{H}^*$,
or filtered back projection, $\vv{W} \vv{H}^*$ or $\vv{H}^* \vv{W} $,
for some appropriate matrix $\vv{W}$.
The same idea can be used with mixtures of linear (or nonlinear) reconstructions,
with the idea that the regressor will learn to spatially vary the regularization strength
based on the image content~\cite{boublil_spatially_2015}.
The main advantage of these methods,
as compared to fully end-to-end methods,
is that the design of the regressor is simplified:
because it always acts on images (rather than vectors of measurements),
it can can employ well-studied images processing tools
such as multiscale processing and convolutions.
This approach has been implemented
with a variety of network architectures
for X-ray \ac{CT}~\cite{jin_deep_2017,chen_low_2017,kang_deep_2018,liu_low_2018}
and \ac{MRI}~\cite{hyun_deep_2018} reconstruction.

\begin{example}[FBPConvNet]
  As an example of some of the practical considerations that go into
  designing and training a learning-based method,
  we present details of one such method,
  FBPConvNet~\cite{jin_deep_2017}, which performs X-ray \ac{CT} reconstruction.
  The FBPConvNet reconstruction method consists of \ac{FBP} followed by a \ac{CNN}.
  The \ac{FBP} part is a standard algorithm which can be performed with MATLAB's
  \texttt{iradon} function.
  On the other hand, special software~\cite{vedaldi_matconvnet_2015} is used to define the \ac{CNN}'s structure
  and train it efficiently.
  The structure of the \ac{CNN} comes from \cite{ronneberger_u_2015},
  with modifications to improve the training stability and performance.
  Designing the architecture of the \ac{CNN} is,
  at this point, an art and requires some trial and error.

  In a simulation experiment,
  the training set consisted of 500 pairs of low-dose measurements
  and the corresponding ground truth reconstruction.
  These data were augmented by flipping each ground truth image in the horizontal and vertical
  directions, resulting in 2,000 training pairs.
  The training itself is a challenging optimization problem
  that is solved using \ac{SGD}---%
  a version of gradient descent that approximates the gradient with a few training samples at a time.
  The \ac{SGD} algorithm requires 
  tuning of several parameters for good convergence.
  Among these are
  the learning rate, which is the step size in gradient descent;
  the learning rate decay, which is a process for the learning rate to decrease after some iterations to improve convergence;
  and 
  the batch size, which is the number of training points to use when computing a stochastic gradient.
  In the case of FBPConvNet, 101 iterations through the training data were sufficient for good performance
  and could be completed in about 15 hours using a \ac{GPU}.
  After training, the algorithm was run on 25 unseen images
  to assess its performance.

  We show the results of this experiment in Figure~\ref{fig:FBPConvNet},
  with a comparison to \ac{FBP} and a \ac{TV}-based approach
  (with regularization parameters chosen to maximize performance on the testing set).
  The results are typical in that the learning-based method
  outperforms the variational one, both quantitative and qualitatively,
  giving slightly sharper-looking reconstructions.
\end{example}

\begin{figure}[!htbp]
  \centering
  \begin{subfigure}{.25\linewidth}
    \includegraphics[width=1\linewidth]{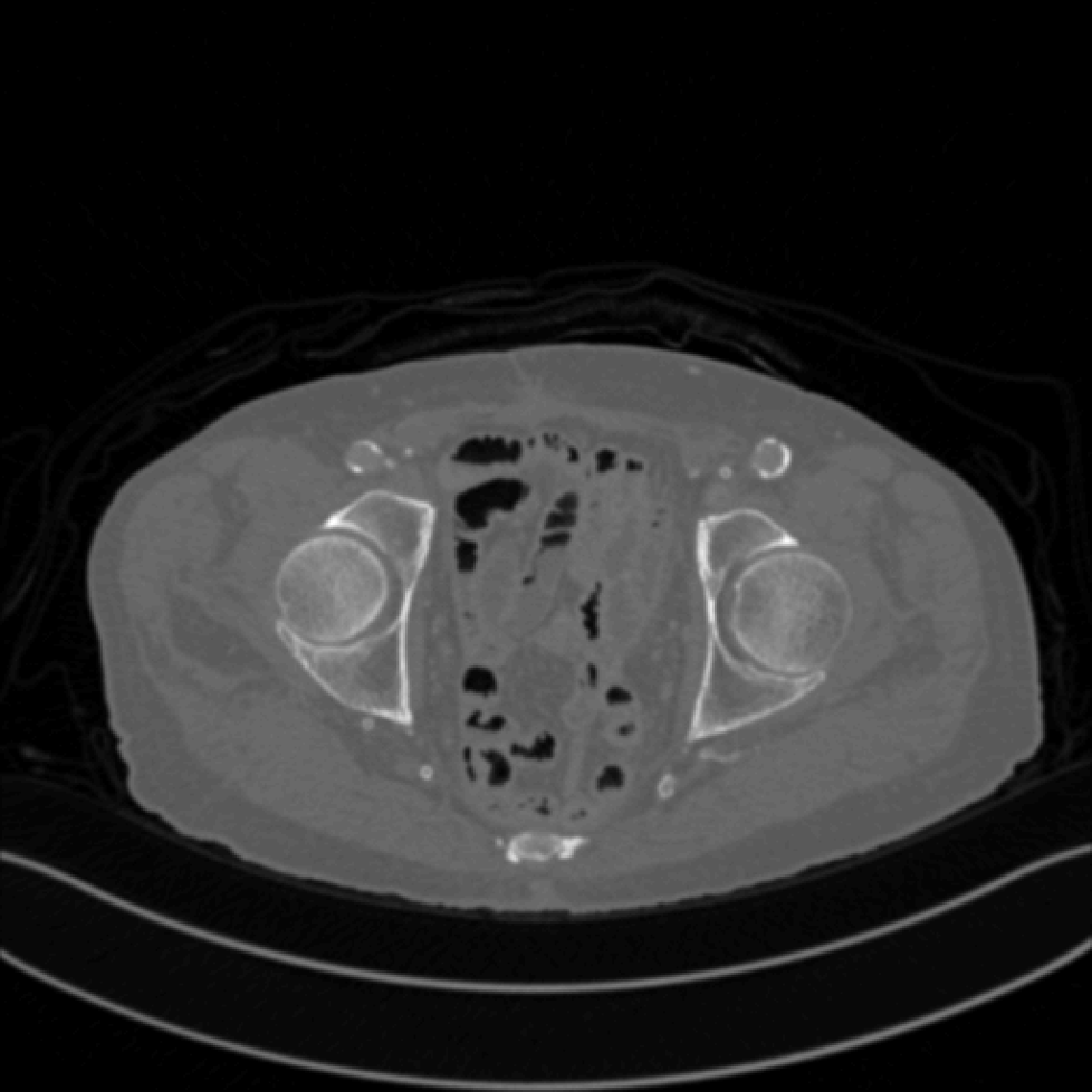}
    \caption{ground truth}
  \end{subfigure}%
  \begin{subfigure}{.25\linewidth}
     \begin{overpic}[width=1\linewidth]{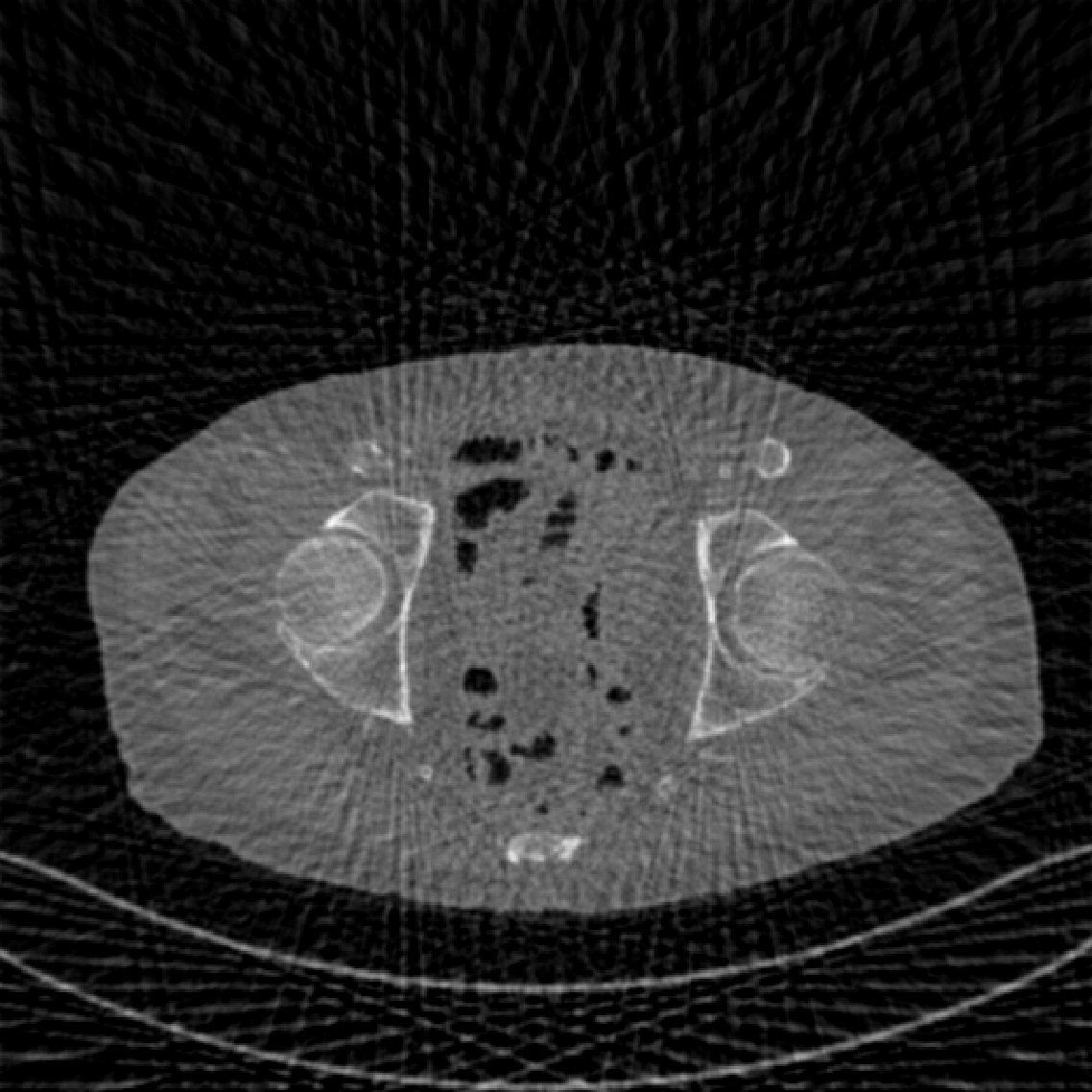}
      \put (2, 90) {\color{white} \footnotesize 13.4 dB}
    \end{overpic}%
    \caption{\ac{FBP}}
  \end{subfigure}%
  \begin{subfigure}{.25\linewidth}
      \begin{overpic}[width=1\linewidth]{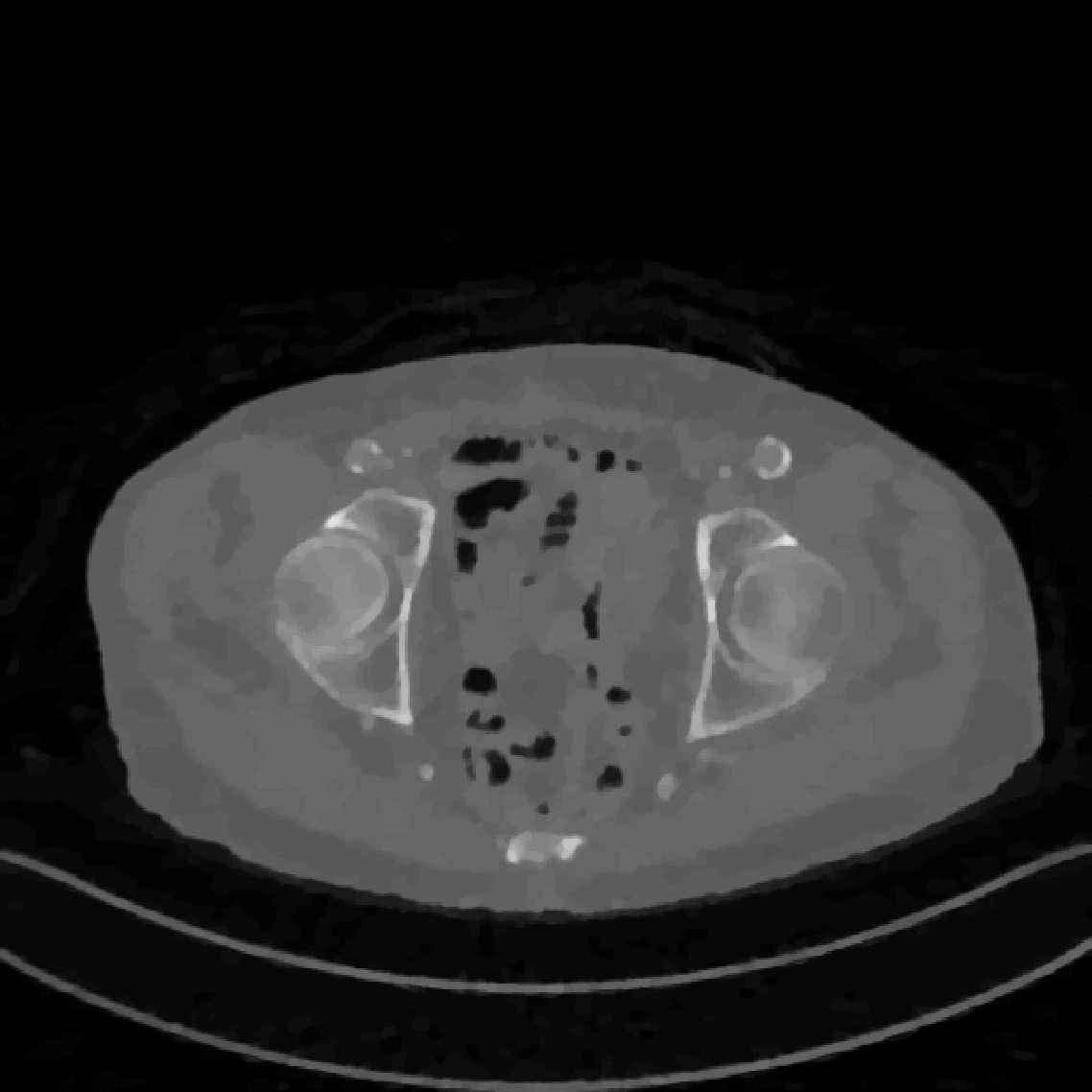}
      \put (2, 90) {\color{white} \footnotesize 24.9 dB}
    \end{overpic}%
    \caption{\ac{TV}}
  \end{subfigure}%
  \begin{subfigure}{.25\linewidth}
     \begin{overpic}[width=1\linewidth]{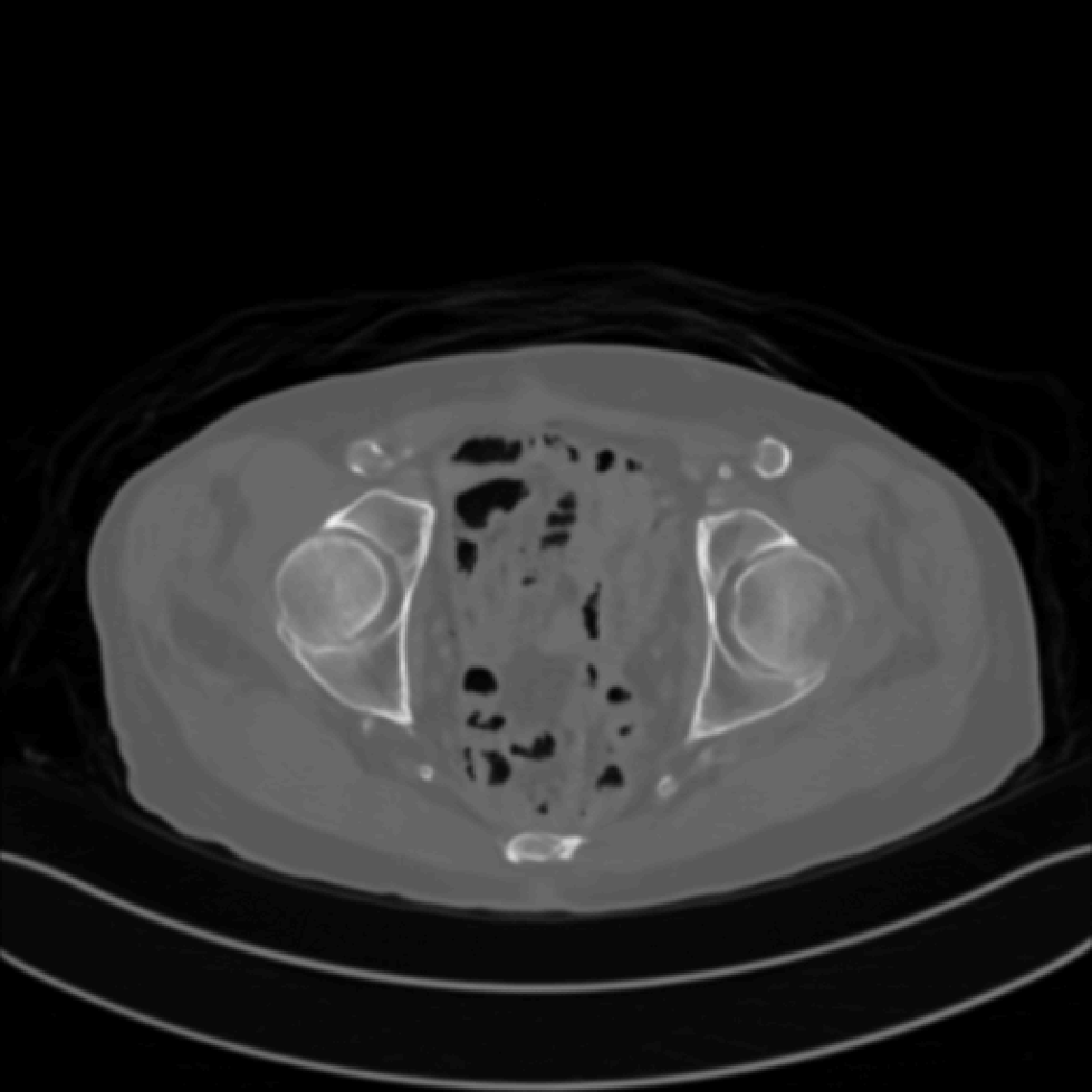}
      \put (2, 90) {\color{white} \footnotesize 28.5 dB}
    \end{overpic}%
    \caption{\ac{FBP} + \ac{CNN}}
   \end{subfigure}
   \caption{Results of a comparison of methods for lose-dose \ac{CT} reconstruction,
     with reconstruction \ac{SNR} given in decibels.
     As is typical of direct methods, \ac{FBP} (b)
     results in artifacts due to the low number of measured views.
     Variational reconstruction with \ac{TV} (c)
     removes many of these artifacts, but at the expense of some oversmoothing.
     Learning-based reconstruction (d) provides sharper-looking reconstructions
     that are quantitatively more accurate than the other methods.
  }
  \label{fig:FBPConvNet}
\end{figure}

A common extension to the learned denoising approach is to also learn a data fidelity term.
This can be done via perceptual loss~\cite{johnson_perceptual_2016},
where images are used as inputs to a pretrained \ac{CNN}
and the distance between them is computed as the
Euclidean distance between their intermediate feature vectors;
the intuition being that a network trained,
e.g., for object recognition,
should building perceptually meaningful intermediate representations
of the images it operates on.
Alternatively, adversarial loss can be employed.
In this case, an image-to-scalar regressor, called a discriminator, is trained
to give low values to real (in the training set) images
and high values to fake (reconstructed images).
The discriminator is trained alongside the reconstruction algorithm,
pushing it to create ever more plausible reconstructions.
See \cite{yang_dagan_2018,quan_compressed_2018,yang_low_2018} for examples
of these approaches for biomedical image reconstruction,
as well as
\cite{zhao_loss_2017} for a comparison of different loss functions
for image restoration.

Finally, we can consider learning the entire reconstruction pipeline end-to-end,
\begin{equation}
   \argmin_{\vv{\theta}} \sum_{t = 1}^{T}
  \| \mathcal{R}_{\vv{\theta}} ( \vv{g}_t ) - \vv{f}_t \|^2_2.
\end{equation}
Here, all of the work goes into designing the structure of $\mathcal{R}_{\vv{\theta}}$
and finding sufficient training data.
One approach is to be problem agnostic,
e.g., using a fully-connected neural network~\cite{durairaj_neural_2007},
or a combination of a fully-connected neural network and a \ac{CNN}~\cite{zhu_image_2018}.
Another is to include some fixed linear layers that are
problem-dependent~\cite{jin_deep_2017,li_learning_2019}.

An alternative design strategy, called unrolling,
is to build the network structure from a preexisting reconstruction algorithm
by turning its iterations into layers in the network.
In effect, the network is a traditional algorithm for solving an inverse problem,
except that some or all of the  linear operations are learned
(rather than being specified by $\vv{H}$)
and the number of iterations/layers is usually low, less than ten,
to prevent the number of parameters in the network from becoming too large.
For example, \cite{jin_deep_2017} notes that
when $\vv{H}^* \vv{H}$ and $\vv{L}$ are both convolutions,
the typical iterate solutions to the $\ell_1$-regularized variational problem
have the form of a \ac{CNN}.
This equivalence allows some learned denoising algorithms
to be interpreted as unrolled iterative methods.
Reference~\cite{kobler_variational_2017} provides another perspective,
proposing to also learn the nonlinear activation functions inside the neural network
as a way to learn the regularization functionals,
and \cite{he_optimizing_2018} explores the idea of parameterizing
the \ac{ADMM} algorithm.
An intermediate approach is to allow $\vv{H}$ to remain fixed,
but to learn, e.g., the gradient step in unrolled gradient descent~\cite{adler_solving_2017a}
or the proximal operators in an unrolled  primal-dual algorithm~\cite{adler_learned_2018}.
Reference \cite{monakhova_learned_2019} explores several unrolled architectures
that vary in terms of how many of the parameters are learned,
concluding that the best performance comes from a balance approach
(i.e., learning some, but not all, of the possible parameters).

The main advantage of learning most or all of the reconstruction pipeline
is that the results can be excellent,
outperforming a well-tuned $\ell_1$ result both quantitatively and qualitatively.
Depending on their specific architecture,
the algorithms can be orders of magnitude faster than iterative methods.
The drawback is that training the algorithms is laborious:
training can take days and is not guaranteed to converge,
meaning that it must be repeated to find suitable parameters
(including training parameters as well as the network architecture itself).
These problems are compounded by the fact that the network may need to be trained again
if aspects of the imaging scenario, e.g., noise level or body part,
change~\cite{knoll_assessment_2018}.
Another drawback is that as more of the reconstruction pipeline is learned,
the data term seems to disappear more and more,
meaning that the returned solution may not be consistent with the measurements.
On one hand, this mismatch may be advantageous:
if the forward model is inaccurate,
enforcing data consistency may make solutions worse (farther from the ground truth).
On the other hand, reconstructions that are plausible-looking
without explaining the measurements are hardly useful images.
One way of enforcing consistency is to mix the learned reconstruction
with a conventional one~\cite{schlemper_deep_2018} %
\begin{equation}
  \argmin_{\vv{f}}\| \vv{g} - \vv{H} \vv{f}\|^2_2 + \lambda \| \vv{f} - \mathcal{R}(\vv{g}) \|^2_2.
\end{equation}

\section{Other Designs}
There are many more designs that do not fit well into the
previous sections.
We share a few here to give a sense of the diversity of the field.

In many problems, we might be able to access high-quality measurements, $\vv{g}_t^{\text{HQ}}$
at the cost of a longer or more-costly acquisition,
e.g., by taking longer X-ray exposures or sampling more of the $k$-space in \ac{MRI}.
The goal is to use these measurements to learn to reconstruct from low-quality (low dose, fast) measurements.
This situation is different from the standard formulation because
now the training measurements are the same as the measurements we aim to reconstruct from.
One solution is simply to generate a training pair $(\vv{f}_t, \vv{g}_t)$
by performing a reconstruction (usually linear) on $\vv{g}_t^{\text{HQ}}$ and by downsampling  $\vv{g}_t^{\text{HQ}}$.
But, $\vv{g}_t^{\text{HQ}}$ can also be used explicitly to learn how to inpaint missing measurements~\cite{eo_kiki_2018}
\begin{equation}
  \argmin_{\vv{\theta}} \sum_{t=1}^{T} \|\mathcal{R}_{\vv{\theta}} ( \mathcal{D}( \vv{g}_t^{\text{HQ}} ) )-  \vv{g}_t^{\text{HQ}} \|^2_2,
\end{equation}
where $\mathcal{D}$ is a model of measurement degradation (e.g., downsampling)
and $\mathcal{R}_{\theta}$ is a parametric inpainting function.
Then, to reconstruct from low-quality measurements, we apply $\mathcal{R}_{\vv{\theta}}$
to generate high-quality measurements and reconstruct using a conventional method.
This is sometimes called data domain learning
and has been used for metal artifact removal in X-ray \ac{CT}~\cite{ghani_fast_2019}.
Another, similar, variation on the setting is learning
to regress from a reconstruction from one measurement type to a reconstruction from another,
for example, to infer one \ac{MRI} scan type from another~\cite{cai_single_2018,kim_improving_2018}.

Another closely-linked topic that does not exactly fit
the paradigm is blind or semiblind image reconstruction.
These involve forward model learning without training data.
Although the problem is superficially impossible---%
recover $\vv{f}$ from $\vv{g} = \vv{H} \vv{f} + \vv{n}$
with unknown $\vv{H}$---%
it can be solved provided enough constraints on $\vv{H}$
and prior knowledge about $\vv{f}$.
The archetypal example is blind deconvolution~\cite{kundur_blind_1996},
where, as with \ac{PSF} estimation from training,
the fact that $\vv{H}$ is a convolution is key.
Many learning algorithms can be adapted to work without training in a similar way,
by alternating a reconstruction step with a learning step that uses the reconstructed images themselves as training;
the challenge is that such a procedure can easily diverge and give nonsensical results.

\section{Where to Get the Training Data}
A key ingredient in any learning method is the training data.
But, how can we acquire training data for inverse problems
without already having a working reconstruction algorithm to make images from measurements?
If we have such a system, why do we need to design another one?
The easiest route out of this chicken-and-egg problem,
as we already mentioned in the previous section,
is to work in the low-quality (few, noisy measurements, sometimes called compressed-sensing) regime,
and to exploit the fact that
we can access high-quality measurements (and reconstructions) at the cost of longer scans.
We can then simulate low-quality measurements by degrading the high-quality ones.
Thus we train a low-quality reconstruction system from a high-quality one.

But, this is not the only way.
In \cite{rivenson_deep_2017},
the authors acquire training data by using
both a low and high numerical aperture microscope,
thus removing the need to simulate low-quality measurements.
The cost of this more-realistic training data is that the
images from both microscopes needs to be correctly registered before learning.
Training data can also come from a highly-accurate physical simulation~\cite{cai_single_2018}
that may be impossible to use directly as a forward model
because of its computational cost or nonlinearity.

Another way is to build algorithms that do not require a training set,
or, at least, not a paired training set
(i.e., measurements and their corresponding reconstructions).
One way to do this is by enforcing cycle consistency~\cite{kang_cycle_2018},
where a pair of algorithms are built, one for reconstruction and one for simulating data.
These algorithms are trained together from unpaired data such that they are inverses.
This approach is similar to generating training from simulation,
except that, here, the simulator itself is learned from data and,
therefore, may be more realistic.

\section{Summary}
Learning-based reconstruction involves using a training set
to tune a parametric reconstruction algorithm.
There is a huge variety in the specific architecture of these algorithms,
with some being closely related to direct or variational methods
and others not.
In any case, creating a suitable training set is a key challenge.

\subsection{Further Reading}
The recent special issue \cite{wang_image_2018}
focuses on original research
on learning-based tomographic reconstruction.
There are several recent reviews around the topic of
learning for image reconstruction:
\cite{mccann_convolutional_2017} focuses on biomedical imaging,
\cite{lucas_using_2018} on general image processing,
\cite{wen_transform_2020} on \ac{MRI} reconstruction,
and \cite{hammernika_machine_2019} on computer vision and medical imaging.
For a long-term perspective on learning in imaging,
see \cite{wang_perspective_2016},
and for a look at how these learning approaches might affect
the practice of radiology, see \cite{zaharchuk_deep_2018}.
For another clinical perspective, \cite{shan_competitive_2019} provides
a double-blind study in which radiologists 
compared learning-based methods to commercial iterative reconstruction methods;
the learning-based methods were shown to be at least as good as the commercial ones.

\chapter{Conclusion}
\label{chap:conclusion}
In this tutorial,
we have given a roughly chronological overview
of biomedical image reconstruction.
We began with a toolbox of mathematical operators
that can be used to build models of many physical imaging systems.
We then showed how to use these forward models
to solve reconstruction problems,
either via direct inversion or variational formulations,
with a focus on the paradigm of sparsity-promoting regularization.
We ended by covering some of the many ways
that training data can be used to develop
new reconstruction methods via machine learning.
Here, we provide a brief comparison of these methods,
commenting on the strength of their theoretical underpinnings,
their complexity, speed, performance, and robustness.
We conclude with comments on future directions for the field.

\section{Comparisons}
Our theoretical understanding is strongest for classical reconstruction methods
and weakest for recent learning-based methods.
The theory behind both classical reconstruction algorithms
and many sparsity-based ones is, for the most part, settled.
Thus, we can often make definitive statements about these algorithms,
e.g., we know that certain iterative procedures converge to global optimums.
We  can sometimes even state that a method is optimal for
a certain class of problems.
That said, there is plenty of work to be done,
especially in pushing our understanding
to ever more complex and realistic data models.
There has been significant progress
in explaining the strong performance
of learning-based methods
on many imaging processing tasks,
including work from the perspective of
approximation theory~\cite{hornik_approximation_1991},
unrolling~\cite{gregor_learning_2010},
and invariants~\cite{mallat_understanding_2016}.
Despite these efforts,
it seems that theory will continue to lag behind practice
for years to come.
We should accept the possibility that
learning approaches may never be understood
as fully as we might like,
and, further, that understanding may become increasingly secondary
to performance.
The fact that the best-understood methods
are not the top performers is not
a new state of affairs in imaging,
but it is understandably disturbing to see the gap widen;
see \cite{elad_deep_2017} for a clear expression of the sentiment.

The trend in complexity among the various reconstruction methods is the same:
classical methods are simple to implement and tune,
while learning-based methods are notoriously hard to train properly.
For classical methods, the number of parameters is small
and can sometimes be calibrated, e.g., by estimating the measurement noise.
Superficially, learning can remove the problem
of parameter selection by fixing parameters during training.
However, training requires its own set of parameters:
the initialization,
gradient step size,
the network architecture, etc.
Because the training problem is usually nonlinear and
nonconvex, all of these parameters affect the result.
However, in a practical settings, such as industrial or medical imaging,
it can be worth investing the engineering effort
to train the system.

In terms of speed,
conventionally we think of trading speed for reconstruction quality:
direct reconstruction is fast but gives the worst results,
increasingly sophisticated regularization improves these results
at the cost of increased runtime.
Learning-based methods upend this order.
Most of the runtime is spent during training,
which happens only once,
after which reconstructions are much quicker.
In fact, this happens because only
network architectures that allow very fast evaluations
can be trained at all.
It is this mode of thinking---%
fix the computational budget, then optimize the algorithm within it---%
that underlies the unrolling concept~\cite{gregor_learning_2010},
and it is a valuable one.

In terms of reconstruction performance,
direct methods are excellent when the number of measurements
is high and the noise is low;
regularization and learning only make sense in the low-dose (high-noise) regime.
(In the full-dose regime,
it is not clear how end-to-end learning can be used at all
because the training data itself would need to come from some reconstruction algorithm.)
In the low-dose regime, however,
regularization or learning is key,
and learning-based methods show the best performance on almost every reconstruction benchmark.
One perspective is that (when sufficient training is available)
learning-based methods provide an upper bound on performance,
which variational methods might achieve if their regularization
and data terms are perfectly suited to the task at hand.
One example of this is in \cite{jin_deep_2017},
where the variational approach beats the learned one
when the data fits the regularization model
(\ac{TV}-based reconstruction of piecewise constant images).

In terms of robustness, learning-based algorithms are currently the weakest.
As we discussed in the last section, slight changes in the imaging parameters
(noise level, body part, etc.) can severely reduce the performance of the network,
or, worse, lead to the appearance of realistic-looking structures in the image
that are not supported by the measurements at all.
Also,
in the context of classification, 
there is the concept of adversarial examples~\cite{szegedy_intriguing_2013},
which are images that are designed to fool a neural network.
While there is debate about why adversarial examples can be created
and what their existence means,
they might indicate that a learning-based image reconstruction algorithm
will occasionally (and inexplicably), give a spurious result.
Recent work has shown examples of this effect in image reconstruction~\cite{antun_instabilities_2020}.
For direct and variational methods, on the other hand,
the reconstruction quality is consistent across datasets
and
(at least in the case of direct and convex variational methods)
degrades gracefully as the measurements degrade.
In particular, the degradation has a predictable form:
noise or specific artifacts (e.g., aliasing in \ac{MRI}, streaks in \ac{CT}).

\section{Future Directions}
The future of the field remains
to produce the highest-quality image using a fixed imaging budget,
e.g., scanning time or radiation dose.
Doing this will continue to require realistic physical models,
though these models may play a new role in reconstruction.
For example, they might be used to create training data,
or they themselves might be learned from training data.
The focus on signal modeling will remain,
and should be pursued both with sophisticated mathematical modeling
and with data-driven approaches.
Finally, we will always seek practical algorithms
with few hand-tuned parameters.

Specifically in the context of biomedical imaging,
we advocate for a shift towards task-oriented evaluations.
The images that we reconstruct in the biomedical domain
are usually used for a specific, predefined purpose,
e.g., to grade a tumor or to measure the quantify the effect
of a genetic manipulation on a model organism.
In this setting, the traditional measures of reconstruction performance
(\ac{SNR}, \ac{MSE}, \ac{SSIM})
are insufficient and should be replaced with task-oriented evaluations.
Such measures may be expensive (if they require human experts)
and they may reveal difficult truths
(e.g., that even a few dB improvement in \ac{SNR} is meaningless),
but they will be critical in creating algorithms
that truly advance the state of biomedical imaging in practice.

Ultimately, future reconstruction algorithms should employ educated learning:
a careful fusion between model-based approaches 
and data-driven ones.
We believe that this approach will provide both excellent performance
and the robustness and performance guarantees needed for biomedical applications.
In biology, this will enable high-quality imaging of a wider range
of biological structures and processes;
in medicine, this will allow diagnostically useful images to be created
more quickly, cheaply, and with less harm to the patient.

\chapter*{Acknowledgements}
\addcontentsline{toc}{chapter}{Acknowledgements}%
This work was made possible thanks to
the CIBM Center for Biomedical Imaging, founded and supported by
Vaud University Hospital Centre (CHUV),
University of Lausanne (UNIL),
Swiss Federal Institute of Technology Lausanne (EPFL),
University of Geneva (UNIGE),
University Hospitals of Geneva (HUG) and the Leenaards and the Louis-Jeantet Foundations
and thanks to the European Research Council through the European Union's Horizon 2020 Research and Innovation Program under Grant 692726, GlobalBioIm.
The authors thank
Philippe Thévenaz,
Julien Fageot,
Avrajit Ghosh,
and
Pol del Aguila Pla.

\backmatter  %

\bibliography{refs_arXiv}
\bibliographystyle{ieeetr}

\end{document}